\documentclass[10pt,final,twocolumn]{IEEEtran}

\usepackage{amsmath,amssymb,amsthm,epsfig,dsfont,color,empheq,enumerate}
\usepackage{algpseudocode, algorithm,epstopdf}
\usepackage{wasysym,epsfig,dsfont,color,graphicx,subfigure}
\usepackage[usenames,dvipsnames]{xcolor}
\usepackage{stfloats,url}
\usepackage{bbold}
\usepackage{bbm}

\newtheorem{proposition}{Proposition}

\theoremstyle{remark}

\newtheoremstyle{remboldstyle}
{}{}{\itshape}{}{\bfseries}{.}{.5em}{{\thmname{#1 }}{\thmnumber{#2}}{\thmnote{ (#3)}}}
\theoremstyle{remboldstyle}

\begin{document}
\title{Data Sketching for Large-Scale Kalman Filtering}

\author{Dimitris Berberidis,~\IEEEmembership{Student Member,~IEEE,}
and Georgios B. Giannakis*,~\IEEEmembership{Fellow,~IEEE}%
\thanks{Work in this paper was supported by the ARO grant no. W911NF-15-1-0492
and NSF grants 1343860, 1442686, and 1500713. }
\thanks{D. Berberidis and G. B. Giannakis are with the ECE Dept., University of Minnesota, Minneapolis, MN 55455, USA. E-mails:\{bermp001,georgios\}@umn.edu}
}

\markboth{IEEE TRANSACTIONS ON SIGNAL PROCESSING (revised, \today)}{Berberidis, Kekatos, and Giannakis: Online Censoring for Large-Scale Regressions with Application to Streaming Big Data}

\maketitle
\begin{abstract}
In an age of exponentially increasing data generation, performing inference tasks by utilizing the available information in its entirety is not always an affordable option. The present paper puts forth approaches to render tracking of large-scale dynamic processes via a Kalman filter affordable, by processing a reduced number of data. Three distinct methods are introduced for reducing the number of data involved in the correction step of the filter. Towards this goal, the first two methods employ random projections and innovation-based censoring to effect dimensionality reduction and measurement selection respectively. The third method achieves reduced complexity by leveraging sequential processing of observations and selecting a few informative updates based on an information-theoretic metric. Simulations on synthetic data, compare the proposed methods with competing alternatives, and corroborate their efficacy in terms of estimation accuracy over complexity reduction. Finally, monitoring large networks is considered as an application domain, with the proposed methods tested on Kronecker graphs to evaluate their efficiency in tracking traffic matrices and time-varying link costs.  

\end{abstract}



\section{Introduction}

Tracking nonstationary dynamic processes is of paramount importance in various applications. In the context of big data, being able to perform accurate and economical state estimation may render problems of prohibitive scale feasible. Weather prediction is an example of tracking a slowly-varying dynamic process, from a massive volume of observations acquired from fast-sampling sensors per time interval; see e.g., \cite{louka2008improvements}. Monitoring large and dynamically evolving networks, where nodes may join or leave and connections may be established or lost as time progresses, provides an exciting domain in which the acquisition and
processing of network-wide performance metrics becomes challenging as the network size increases~\cite[Ch. 8]{kolaczyk2014statistical}. For instance, monitoring path metrics
such as delays or loss rates is challenging primarily because
the number of paths generally grows as the square of the
number of nodes in the network. Therefore, measuring and
storing the delays of all possible origin-destination pairs is
hard in practice, even for moderate-size networks~\cite{kolaczyk2014statistical}.  

In this context, efforts to reduce the number of measurements used for tracking have pursued two different directions. The first is that of optimal experimental design (OED), where the goal is to perform \emph{model-driven} sensor selection based on ensemble performance metrics (e.g. the trace of the error covariance). Channel-aware dimensionality reduction of observations was reported in~\cite{zhu2009power} and~\cite{ma2014distributed} using distributed wireless sensor networks (WSNs). Optimal and near optimal sensor schedules for a finite time horizon estimation was dealt with in~\cite{vitus2012efficient}, while entropy- and mutual-information-based sensor selection were advocated in~\cite{wang2004entropy} and~\cite{ertin2003maximum}. A posterior-CRLB-based method to select sensors for tracking was introduced in~\cite{zuo2007posterior}, via convex optimization in~\cite{boyd} and~\cite{altenbach2012strategies}, while a greedy algorithm leveraging submodularity was developed in~\cite{shamaiah2010greedy} for measurement selection in sequential estimation. The latter has also been advocated as a means of reducing the complexity of Kalman filters that operate with limited processing resources~\cite{anytimeKF}. OED is nicely attuned for designing low-dimensional observation models, but it is \emph{data-agnostic} and thus sub-optimal when observations become available and need to be reduced.
 
The second direction is that of \emph{data-driven} methods that select available measurements for processing. 
Specifically, censoring has recently been employed to select data for distributed parameter estimation using resource-constrained WSNs, thus trading off performance for tractability~\cite{tsp2012ggeric,tsp2013youxie}. Furthermore, censoring has been proposed for signal estimation using WSNs, for tracking, and control of dynamical processes~\cite{liu2014survey,battistelli2012data,zheng2014sequential,arxiv2014wang}. However, existing works on censoring mainly aim at reducing the rate at which sensors communicate their observations, and pertinent methods exhibit large computational complexity and storage requirements, which can be possibly afforded only at the fusion center. 

The goal of this paper is to perform reliable tracking using the Kalman filter (KF), while reducing the amount of data and the computational complexity involved. To this end, the first two methods employ random projections and innovation-based censoring to respectively effect dimensionality reduction and measurement selection. The third method reduces complexity by leveraging sequential processing of observations when the noise is uncorrelated, and by selecting a few informative updates based on an information-theoretic metric. Finally, an efficient backward smoothing method is developed to mitigate the performance degradation caused by dimensionality reduction.  Corroborating simulations compare with state-of-the-art greedy measurement selection algorithms, and illustrate the efficacy of the novel schemes. To demonstrate the applicability of the proposed update selection approach on real-world problems, traffic matrix estimation and network link cost estimation is also considered.  

The rest of the paper is organized as follows. Section \ref{sec:problem} introduces the proposed model of reduced complexity KF. Sections \ref{sec:rpkf} and \ref{sec:AC-KF} present the two dimensionality reduction modules based on RPs and censoring, respectively. The proposed update selection method is introduced in Section \ref{sec:rc_KF}. Numerical experiments are in Section~\ref{sec:simulations}, while Section~\ref{sec:networks} includes experiments on network monitoring. Finally, concluding remarks are given in Section~\ref{sec:conclusion}.  

\emph{Notation.} Lower- (upper-) case boldface letters denote column vectors (matrices). Calligraphic symbols are reserved for sets, while $^T$ stands for transposition. Vectors $\mathbf{0}$, $\mathbf{1}$, and $\mathbf{e}_n$ denote the all-zeros, the all-ones, and the $n$-th canonical vector, respectively. Symbol $\mathds{1}_{E}$ denotes the indicator for the event $E$. Notation $\mathcal{N}(\mathbf{m},\mathbf{C})$ stands for the multivariate Gaussian distribution with mean $\mathbf{m}$ and covariance matrix $\mathbf{C}$, while $\mathrm{tr(\mathbf{X})}$, $\lambda_{\min}(\mathbf{X})$, and $\lambda_{\max}(\mathbf{X})$ are reserved for the trace, the minimum and maximum eigenvalues of matrix $\mathbf{X}$, respectively.  {Symbol $D\gg$ is used to denote that the number of observations $D$ is ``prohibitively large'' relative to the problem at hand, as well as the computing platform.}

\section{Problem Statement and Preliminaries}\label{sec:problem}
Consider the following linear dynamical system model
\begin{align}\label{model2a}
	\boldsymbol{\theta}_n&= \mathbf{F}_n\boldsymbol{\theta}_{n-1}+\mathbf{G}_n\mathbf{u}_n+\mathbf{w}_n\\
	\mathbf{y}_n&= \mathbf{X}_n\boldsymbol{\theta}_n + \mathbf{v}_n \label{model2b}
\end{align}
where $\boldsymbol{\theta}_n\in\mathbb{R}^{p}$ denotes the state vector at time $n$; $\mathbf{F}_n$ is the known state-transition matrix; $\mathbf{G}_n$ and $\mathbf{u}_n$ are known, deterministic control-input model and control-input vector respectively; $\mathbf{y}_n\in\mathbb{R}^{D}$ the measurement vector, and $\mathbf{X}_n$ is the known $D\times{p}$ measurement matrix; while $\mathbf{w}_n$ and $\mathbf{v}_n$ are zero-mean, mutually uncorrelated and individually uncorrelated across time random noise vectors, with respective covariance matrices $\mathbf{Q}_n$ and $\mathbf{R}_n$. The initial state $\boldsymbol{\theta}_0$ has mean $\mathbf{m}_0$, and covariance $\mathbf{P}_0$.
\begin{figure}[t!]
	\centering
	\centerline{\includegraphics[width=1\linewidth, height=1 in]{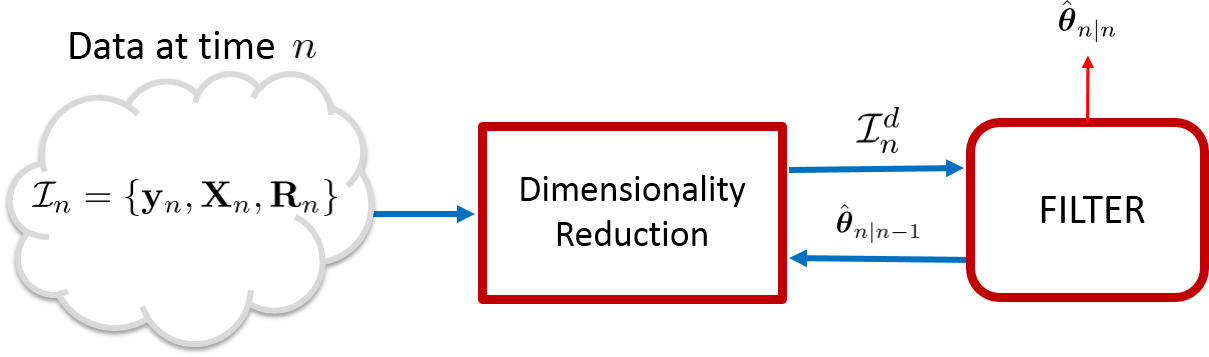}}
	\caption{Reduced-dimension filtering.}
	\label{fig:DR}
	\vspace{-0.2 cm}
\end{figure}

Given the information-bearing data $\mathcal{I}_n:=\{\mathbf{y}_n,\mathbf{X}_n,\mathbf{R}_n\}$ of the measurement~\eqref{model2b} at time $n$, the most recent estimate $\hat{\boldsymbol{\theta}}_{n-1|n-1}$ and its covariance matrix $\mathbf{P}_{n-1|n-1}$, the celebrated KF yields the minimum mean-square error (MMS- E) optimal estimate $\hat{\boldsymbol{\theta}}_{n|n}$ in two steps. First, the state prediction $\hat{\boldsymbol{\theta}}_{n|n-1}$ and its covariance matrix $\mathbf{P}_{n|n-1}$ are obtained using the model dynamics $\{\mathbf{F}_n, \mathbf{Q}_n \}$ as [cf.~\eqref{model2a}]
\begin{subequations}
	\label{pred}
	\begin{align*}
		\hat{\boldsymbol{\theta}}_{n|n-1}&=\mathbf{F}_n\hat{\boldsymbol{\theta}}_{n-1|n-1}+\mathbf{G}_n\mathbf{u}_n\\
		\mathbf{P}_{n|n-1}&=\mathbf{F}_n\mathbf{P}_{n-1|n-1}\mathbf{F}_n^T+\mathbf{Q}_{n}.
	\end{align*}
\end{subequations}
Subsequently, as $\mathcal{I}_n$ becomes available, $\hat{\boldsymbol{\theta}}_{n|n}$ is obtained as   
\begin{align}\label{costf}
	\hat{\boldsymbol{\theta}}_{n|n}=\arg\min_{\boldsymbol{\theta}}\|\mathbf{y}_n-\mathbf{X}_n\boldsymbol{\theta}\|_{\mathbf{R}_n^{-1}}^2+\|\boldsymbol{\theta}-\hat{\boldsymbol{\theta}}_{n|n-1}\|_{\mathbf{P}_{n|n-1}^{-1}}^2.
\end{align}
The first term of the cost in~\eqref{costf} is a weighted least-squares term fitting the state $\boldsymbol{\theta}$   with $\mathcal{I}_n$ that arises from the linear observation model in~\eqref{model2b}; while the second regularization term corresponds to treating $\hat{\boldsymbol{\theta}}_{n|n-1}$ as a prior of ${\boldsymbol{\theta}}_{n}$. Solving~\eqref{costf} and applying the matrix inversion lemma (MIL) yields the well known KF correction step, e.g.,~\cite[p. 205]{bar2004estimation}
\begin{align*}
	\hat{\boldsymbol{\theta}}_{n|n}&=\hat{\boldsymbol{\theta}}_{n|n-1}+\mathbf{K}_n({\mathbf{y}}_n-{\mathbf{X}}_n\hat{\boldsymbol{\theta}}_{n|n-1})
\end{align*}
where the so-termed KF gain $\mathbf{K}_n$ and the state covariance update are given by
\begin{align*}
	\mathbf{K}_n&=\mathbf{P}_{n|n-1}{\mathbf{X}}_n^T\left({\mathbf{X}}_n\mathbf{P}_{n|n-1}{\mathbf{X}}_n^T+{\mathbf{R}}_n\right)^{-1}\\
	\mathbf{P}_{n|n}&=\left(\mathbf{I}_p-\mathbf{K}_n{\mathbf{X}}_n\right)\mathbf{P}_{n|n-1}.
\end{align*}
A dual form of the KF known as the information filter (IF) relies on the MIL to offer a more efficient solver of~\eqref{costf} as $D$ grows large~\cite[Ch. 7]{bar2004estimation}. Nevertheless, even the low-complexity IF requires $\mathcal{O}(Dp^2)$ multiplications to solve~\eqref{costf} in the case of uncorrelated observations ($\mathbf{R}_n$ diagonal), and $\mathcal{O}(D^2p)$ in general. Therefore, for large-scale KF problems where $D\gg$, dimensionality reduction of the datasets $\mathcal{I}_n$ is well motivated for rendering the solution of~\eqref{costf} computationally tractable, while also reducing other data-related costs, such as storage and transmission.

Towards this goal, we introduce a reduced-complexity Kalman-like filter (see Algorithm~\ref{reducedKF}) that extracts a reduced (size $d<D$), yet informative dataset $\mathcal{I}_n^d:=\{\check{\mathbf{y}}_n,\check{\mathbf{X}}_n,\check{\mathbf{R}}_n\}$ from the original $\mathcal{I}_n$, where $\check{\mathbf{y}}_n\in\mathbb{R}^d,\check{\mathbf{X}}_n\in\mathbb{R}^{d\times{p}}$ and $\check{\mathbf{R}}_n\in\mathbb{R}^{d\times d}$ are the corresponding reduced-dimension observation vector, measurement matrix, and covariance matrix; see also Fig.~\ref{fig:DR}. Consequently, the problem reduces to the design of low-complexity sketching modules for informative dimensionality reduction. In the ensuing two sections, a \emph{data-agnostic} method based on RPs followed by a \emph{data-adaptive} method based on censoring are developed.

\begin{algorithm}[t!]
	\caption{Reduced-dimension KF}\label{reducedKF}
	\begin{algorithmic}
		\State \underline{\textbf{Initialization:}} $\hat{\boldsymbol{\theta}}_{0|0}=\mathbf{m}_0,~\mathbf{P}_{0|0}=\mathbf{P}_0$
		\For {$n=1:N$}
		\State \underline{\textbf{Prediction Step}}
		\State \hspace{0.07 cm}$\hat{\boldsymbol{\theta}}_{n|n-1}=\mathbf{F}_n\hat{\boldsymbol{\theta}}_{n-1|n-1}+\mathbf{G}_n\mathbf{u}_n$
		\State $\mathbf{P}_{n|n-1}=\mathbf{F}_n\mathbf{P}_{n-1|n-1}\mathbf{F}_n^T+\mathbf{Q}_{n}$
		
		\State \underline{\textbf{Data Reduction}}
		\State $\{\check{\mathbf{y}}_n,\check{\mathbf{X}}_n,\check{\mathbf{R}}_n\}=\mathrm{Sketching}\big(\{\mathbf{y}_n,\mathbf{X}_n,\mathbf{R}_n\},~\hat{\boldsymbol{\theta}}_{n|n-1} \big)$
		
		\State \underline{\textbf{Correction Step}}
		\State \hspace{0.07 cm}$\hat{\boldsymbol{\theta}}_{n|n}=\hat{\boldsymbol{\theta}}_{n|n-1}+\check{\mathbf{K}}_n(\check{\mathbf{y}}_n-\check{\mathbf{X}}_n\hat{\boldsymbol{\theta}}_{n|n-1})$
		\State \hspace{0.2 cm}$\check{\mathbf{K}}_n=\mathbf{P}_{n|n-1}\check{\mathbf{X}}_n^T\left(\check{\mathbf{X}}_n\mathbf{P}_{n|n-1}\check{\mathbf{X}}_n^T+\check{\mathbf{R}}_n\right)^{-1}$
		\State $\mathbf{P}_{n|n}=\left(\mathbf{I}_p-\check{\mathbf{K}}_n\check{\mathbf{X}}_n\right)\mathbf{P}_{n|n-1}$
		\EndFor
	\end{algorithmic}
\end{algorithm}

\begin{algorithm}[t!]%
	\vspace{0.5 cm}
	\caption{RP sketching module}\label{RPKF}
	\begin{algorithmic}
		\State {Dimensionality reduction with RPs}
		\State \hspace{0.1 cm}$\check{\mathbf{y}}_n=\mathbf{S}_d\mathbf{H\Lambda y}_n$
		\State $\check{\mathbf{X}}_n=\mathbf{S}_d\mathbf{H\Lambda X}_n$
		\State $\check{\mathbf{R}}_n=\mathbf{S}_d\mathbf{H\Lambda R}_n(\mathbf{S}_d\mathbf{H\Lambda})^T	$	
	\end{algorithmic}
\end{algorithm}

\section{RP-based KF}\label{sec:rpkf}

RP-based dimensionality reduction amounts to premultiplying measurements and regressors $\{\mathbf{y}_n, \mathbf{X}_n\}$ with a random matrix $\mathbf{H}$, and a diagonal matrix $\mathbf{\Gamma}$, whose entries take the values $\{+1/\sqrt{D},-1/\sqrt{D}\}$ equiprobably. The net result is a linear transformation of the measurement equations so that all rows convey ``comparable information''. A subset of $d$ rows of the transformed system is then extracted by simple random sampling, implemented by left multiplication with a random $d\times{D}$ selection matrix $\mathbf{S}_d$. 

Originally developed for linear regressions~\cite{mahoney2011ftml,dmmm06acm,woodruff}, the novelty here is RP-based reduced-dimensionality tracking of dynamical processes. Applying the Hadamard preconditioning and random sampling matrices on~\eqref{model2b} yields the reduced-dimension observation model        
\begin{equation*}
\check{\mathbf{y}}_n:=\mathbf{S}_d\mathbf{H}\mathbf{\Gamma}\mathbf{y}_n=\mathbf{S}_d\mathbf{H}\mathbf{\Gamma}(\mathbf{X}_n\boldsymbol{\theta}_n+\mathbf{v}_n)=\check{\mathbf{X}}_n\boldsymbol{\theta}_n+\check{\mathbf{v}}_n
\end{equation*}
where $\check{\mathbf{v}}_n:=\mathbf{S}_d\mathbf{H}\mathbf{\Gamma}\mathbf{v}_n$ is zero mean with covariance $\check{\mathbf{R}}_n=\mathbf{S}_d\mathbf{H}\mathbf{\Gamma}\mathbf{R}_n(\mathbf{S}_d\mathbf{H}\mathbf{\Gamma})^T$. Given $\hat{\boldsymbol{\theta}}_{n|n-1}$ and the reduced data $\mathcal{I}_n^d$, state estimate $\hat{\boldsymbol{\theta}}_{n|n}$ can be obtained as [cf. ~\eqref{costf}] 
\begin{align}\label{costfrp}
\hat{\boldsymbol{\theta}}_{n|n}=\arg\min_{\boldsymbol{\theta}}\|\check{\mathbf{y}}_n-\check{\mathbf{X}}_n\boldsymbol{\theta}\|_{\check{\mathbf{R}}_n^{-1}}^2+\|\boldsymbol{\theta}-\hat{\boldsymbol{\theta}}_{n|n-1}\|_{\mathbf{P}_{n|n-1}^{-1}}^2.
\end{align}
Solving~\eqref{costfrp} and applying the MIL  yields the novel RP-based KF, which is summarized as Algorithm~\ref{reducedKF} using Algorithm~\ref{RPKF} as sketching module. 

Implementing RPs can have affordable $\mathcal{O}(Dp\log d)$ complexity if $\mathbf{H}$ is chosen to be a pseudo-random Hadamard matrix of size $D=2^{\nu}$ for $\nu\in\mathbb{Z}_+$. Different from the more elaborate approaches in \cite{zhu2009power} and \cite{ma2014distributed}, the proposed RP-KF is an easy-to-implement, ``one-size-fits-all'' reduced-complexity tracker, using data-agnostic dimensionality reduction. Furthermore, RP-KF's estimation performance can be guaranteed as asserted in the ensuing proposition, which provides a benchmark for the data-driven methods introduced in the following section.  
\begin{proposition}\label{pro:rps}
	{With $\mathbf{R}_n=\sigma_n^2\mathbf{I}_D$, let $\mathbf{A}_n:=[\mathbf{P}_{n|n-1}^{-1/2},$ $~~ \sigma_n^{-1}\mathbf{X}_n^T]^T$, $\mathbf{b}_n:=[\mathbf{P}_{n|n-1}^{-1/2}\hat{\boldsymbol{\theta}}_{n|n-1},~~ \sigma_n^{-1}\mathbf{y}_n^T]^T$, and $\mathbf{A}_n=\mathbf{U}_n{\mathbf{\Lambda}_n}\mathbf{V}_n^T$ the singular value decomposition of $\mathbf{A}_n$}. If $\|\mathbf{U}_n\mathbf{U}_n^T\mathbf{b}_n\|_2\geq{\gamma\|\mathbf{b}_n}\|_2$ for some $\gamma \in (0,1]$, then by choosing $d=\mathcal{O}(p\ln(pD)/\epsilon)$ the following bound for the RP-KF estimates holds w.h.p. 
	\begin{align*}
	\|\hat{\boldsymbol{\theta}}_{n|n}-\hat{\boldsymbol{\theta}}_{n|n}^{\star}\|_2\leq\sqrt{\epsilon}\left(\kappa(\mathbf{A}_n)\sqrt{\gamma^{-2}-1}\right)\|\hat{\boldsymbol{\theta}}_{n|n}^{\star}\|_2
	\end{align*}
	where $\kappa(\mathbf{A}_n)$ denotes the condition number of $\mathbf{A}_n$, and $\hat{\boldsymbol{\theta}}_{n|n}^{\star}$ is the full-data KF estimate.
\end{proposition}
\begin{proof}
	See Appendix 1.
\end{proof}
Proposition~\ref{pro:rps} asserts that, per time slot $n$, the estimate $\hat{\boldsymbol{\theta}}_{n|n}$ of the RP-KF can be guaranteed to be close enough, in the relative squared-error sense, to the estimate $\hat{\boldsymbol{\theta}}_{n|n}^{\star}$ of the full KF, if the reduced dimension $d$ is chosen to be large enough. {Note that Proposition 1 only provides a \emph{per-step} error guarantee, meaning that the error between the estimates of the RP-KF and the full-data KF at slot $n$ is bounded, given that the two filters share a common estimate at slot $n-1$. Bounding the RP-KF error across multiple time slots is a more challenging task that goes beyond the scope and claims of the present paper.} Naturally, the quality of the approximation also depends on other parameters such as the observation matrix, noise variance and covariance of prediction. Nevertheless, being data-agnostic and requiring storage and processing of $\mathcal{I}_n$ in batch form per time slot renders the RP-KF less attractive in practice, and motivates the algorithms presented in the following two sections. 

\section{CENSORING-BASED KF}\label{sec:AC-KF}

Measurement censoring for estimating dynamical processes has been advocated as a means of reducing the inter-sensor transmission overhead when WSNs are deployed for distributed tracking~\cite{battistelli2012data,zheng2014sequential}; see also~\cite{liu2014survey,wang1997event}, where censoring is employed for event-based estimation. Since the goal in the aforementioned applications is saving communication resources, censoring is performed solely on measurements $\mathbf{y}_n$, with $\mathbf{X}_n$ and $\mathbf{R}_n$ assumed known and used even for the censored entries of $\mathbf{y}_n$; {thus,~\cite{battistelli2012data,liu2014survey,wang1997event,tsp2013youxie}, and \cite{zheng2014sequential} rely on reducing the dimensionality of a dataset that only consists of observations; that is, $\mathcal{I}_n:=\{\mathbf{y}_n\}$}. A subset of $d$ observations $\mathcal{I}^d_n:=\{\mathbf{y}_{n,{\mathcal{S}_n}}\}$ is obtained, where $y_{n,i}$ is the $i-$th entry of $\mathbf{y}_n$, and $\mathcal{S}_n\subseteq\{1,\ldots,D\}$  denotes a set collecting the indices of uncensored observations. Given $\mathbf{y}_{n,{\mathcal{S}_n}}, \mathbf{X}_n$ and $\mathbf{R}_n$,~\cite{battistelli2012data,liu2014survey,wang1997event,tsp2013youxie,zheng2014sequential} develop sequential estimators to optimally estimate $\boldsymbol{\theta}_n$. Targeting reduction of communication load, optimal (in the maximum likelihood or MMSE sense) estimation from (un)censored observations comes with complexity comparable to that of using the full set of $D$ measurements. For our big-data setups, this is not affordable.

Since the aim is dimensionality and complexity reduction, the starting point is on censoring entire rows of the full dataset $\mathcal{I}^D_n:=\{\mathbf{y}_n,\mathbf{X}_n,\mathbf{R}_n\}$, in order to obtain a reduced set $\mathcal{I}^d_n:=\{\mathbf{y}_{n,{\mathcal{S}_n}},\mathbf{X}_{n,{\mathcal{S}_n}},\mathbf{R}_{n,{\mathcal{S}_n}}\}$, where $\mathbf{x}_{n,i}^T$ denotes the $i-$th row of $\mathbf{X}_n$ and $\mathbf{R}_{n,{\mathcal{S}_n}}:=\mathrm{cov}(\mathbf{v}_{n,{\mathcal{S}_n}})$. The goal here is to develop censoring rules in order to obtain $\mathcal{S}_n$, so that $\mathcal{I}^d_n$ is an ``informative'' subset of $\mathcal{I}_n$. Most existing censoring schemes adopt the innovation $\tilde{\mathbf{y}}_{n}:=\mathbf{y}_n-\mathbf{X}_n\hat{\boldsymbol{\theta}}_{n|n-1}$ as a measure of information contained in $\mathbf{y}_n$.

 One approach --henceforth termed \emph{block censoring} (BC)-- is to censor the entire vector $\mathbf{y}_n$. From an information-theoretic viewpoint~\cite{zheng2014sequential}, the optimal BC rule relies on the magnitude of the prewhitened innovation $\boldsymbol{\Sigma}^{-1/2}_{n}\tilde{\mathbf{y}}_{n}$, where $\boldsymbol{\Sigma}_n:=\mathrm{cov}(\tilde{\mathbf{y}}_{n})=\mathbf{X}_n\mathbf{P}_{n|n-1}\mathbf{X}_n^T+\mathbf{R}_n$; thus, $\mathcal{S}_n$ is found as 
\begin{align}\label{block_censoring}
\mathcal{S}_n:=\bigg\{
\begin{array}{lc}
\{1, \ldots, D\},& \|\boldsymbol{\Sigma}^{-1/2}_{n}\tilde{\mathbf{y}}_{n}\|_2>\tau_n \\
\emptyset,& \mathrm{otherwise} 
\end{array}.
\end{align} 
Clearly, having $\mathcal{S}_n=\emptyset$ corresponds to skipping the correction step of the KF. A major shortcoming of~\eqref{block_censoring} is the cubic complexity $\mathcal{O}(D^3)$ associated with inverting $\boldsymbol{\Sigma}_{n}$. Furthermore, BC-KF can only reduce the data cost \emph{on average} across iterations by entirely skipping correction steps.

Our idea of a more attractive alternative is to possibly censor separately each entry of $\mathcal{I}_n$. 
Such an entry-wise censoring rule yields $\mathcal{S}_{n}$ as
{\begin{align}\label{entry_censoring}
\mathcal{S}_n:=\{ 1\leq i\leq D~ \big{|}~|\sigma_{n,i}^{-1}\tilde{{y}}_{n,i}|>\tau_n\}
\end{align}  
where $\sigma_{n,i}^2=[\mathbf{R}_{n}]_{ii}$, and $\tau_n$ can be tuned so that the set cardinality $|\mathcal{S}_n|\approx{d}$. Note that normalization in (6) includes pre-whitening of the innovation vector that is effected through left multiplication with $\mathbf{\Sigma}_n^{-1/2}$. In contrast, any low-complexity \emph{per-entry adaptive} censoring rule  cannot explicitly consider the cross-correlation between different measurements, meaning that $\mathbf{\Sigma}_n$ is not utilized. Instead, \eqref{entry_censoring} uses the diagonal entries of $\mathbf{R}_n$ as an approximate measure of the per-entry innovation variance and therefore it serves as a normalization factor.  Compared to BC-KF, the innovation-based entry-wise rule of~\eqref{entry_censoring} is more flexible in reducing the available data, since it can censor any subset of observations at slot $n$ at much lower complexity. Nevertheless, to accurately perform measurement selection with \eqref{entry_censoring}, $|\tilde{{y}}_{n,i}|$ must reflect how informative ${{y}}_{n,i}$ is for the purpose of tracking $\boldsymbol{\theta}_{n}$. Using for this purpose the entry-wise predictor-based innovations $\tilde{y}_{n,i}:=y_{n,i}-\mathbf{x}_{n,i}^T\hat{\boldsymbol{\theta}}_{n|n-1}$ is a possibility, but turns out to be unsuitable for the proposed reduced-complexity KF, due to the fact that censoring rule \eqref{entry_censoring} tends to yield ``biased'' observations for a given  $\hat{\boldsymbol{\theta}}_{n|n-1}$. Correction of this bias is possible through the incorporation of a maximum likelihood criterion (see, e.g. \cite{BerberKekatosTSP}). Since such an approach requires the additional knowledge of $\mathbf{X}_n$ and incurs computational complexity at least as high as that of the full-data KF, it is only suitable for reducing the communication overhead. }

Targeting a more suitable censoring rule, the \emph{adaptive censoring} least mean-square (AC-LMS) algorithm we introduced in \cite{BerberKekatosTSP} for non-dynamical regressions can be employed to discard uninformative rows of $\mathcal{I}_n$. Within time slot $n$, rows of $(\mathbf{y}_n, \mathbf{X}_n)$ are processed sequentially; given a temporary estimate $\hat{\boldsymbol{\theta}}_{n|n-1,i}$, the $i-$th row is discarded when indicated so by the censoring variable ($\mathds{1}_{.}$ denotes the indicator function)
{\begin{equation}\label{cis} 
c_i:=\mathds{1}{\{|{y_{n,i}-\mathbf{x}_{n,i}^T\hat{\boldsymbol{\theta}}_{n|n-1,i-1}}|\leq \tau_n\sigma_{n,i}^{-1}\}}.
\end{equation}}
Given $\{ \mathbf{x}_{n,i}^T, \hat{\boldsymbol{\theta}}_{n|n-1,i-1}, \tau_n\}$, a ``censoring slab'' is specified in $\mathbb{R}^{p+1}$ to determine whether $(\mathbf{x}_{n,i},y_{n,i})$ will be censored (if inside this slab) or not (if outside this slab). 
 If deemed informative enough $(c_i=0)$, the $i-$th row is added to $\mathcal{S}_n$, and subsequently involved in updating $\hat{\boldsymbol{\theta}}_{n|n-1,i}$ as 
\begin{equation}\label{update} 
\hat{\boldsymbol{\theta}}_{n|n-1,i}=\hat{\boldsymbol{\theta}}_{n|n-1,i-1}
+(1-c_i)\mu \mathbf{x}_{n,i} \left({y}_{n,i}-\mathbf{x}_{n,i}^T\hat{\boldsymbol{\theta}}_{n|n-1,i-1}\right).
\end{equation}
The role of the first-order update in~\eqref{update} is to perturb the censoring slab towards the direction of $(\mathbf{x}_{n,i},y_{n,i})$, thus making it less likely for future measurements conveying information ``close to'' $(\mathbf{x}_{n,i},y_{n,i})$ to be retained. {Intuitively speaking, such updates eliminate measurement redundancies and reduce estimation error due to the bias of uncensored observations.} The AC sketching module is summarized as Algorithm~\ref{ACensor}, and when plugged into Algorithm~\ref{reducedKF}, it yields the proposed adaptive censoring (AC)-KF scheme. With regards to its performance, we have the following result.
{\begin{proposition}\label{pro:bias}
   If $\mathbf{R}_n=\sigma_n^2\mathbf{I}$, $\sigma_n^2\gg\mathrm{tr}(\mathbf{P}_{n|n-1})$, and $\mathbf{w}_n$, $\mathbf{v}_n$ are zero-mean Gaussian, then the AC-KF with $\mu=0$ yields unbiased estimates $\forall{\tau}$. 
\end{proposition}}
\begin{proof}
	See Appendix 3.
\end{proof}
The assumptions in Proposition~\ref{pro:bias} were made to simplify the proof and are not necessary. Extensive simulations indicate that the AC-KF remains unbiased even for low $\sigma_n^2$ and correlated noise, and also for $\mu>0$. Nevertheless, the variance of AC-KF largely depends on the choice of $\mu$. Tuning $\mu$ to optimize the MSE performance of AC-KF is a challenging task. Accurate rules for selecting $\mu$ is part of our ongoing research. However, even for possibly suboptimal values of $\mu$, the proposed scheme yields promising results. Simulations in Section~\ref{sec:simulations} will demonstrate that the proposed AC-KF attains estimation accuracy close to that of the KF using the greedy measurement selection method in~\cite{shamaiah2010greedy}. In addition, the proposed sketching module performs a single pass over the data, and requires $\mathcal{O}(Dp)$ computations, which is markedly lower than the $\mathcal{O}(Ddp^2)$ required to perform greedy selection. Furthermore, AC-KF is suitable for online implementation by processing rows of $\mathcal{I}_n$ sequentially. {Table \ref{tab:table_1} summarizes the per-slot computational complexity of applying the sketching modules corresponding to AC-KF, RP-KF and random sampling.}     

{
	\textbf{Remark 1:} Note that (9) does not pertain to a filter update. Instead, it is an update of  $\hat{\boldsymbol{\theta}}_{n|n-1,i}$ within the AC sketching module (Algorithm~\ref{ACensor}) that leverages uncensored entries per slot, and it is only used for censoring entries within $\mathcal{I}_n$ of slot $n$. Note also that for $\mu >0$ Algorithm~\ref{ACensor} performs \emph{adaptive} censoring on $\mathcal{I}_n$. If $\mu=0$, then $\hat{\boldsymbol{\theta}}_{n|n-1,i}=\hat{\boldsymbol{\theta}}_{n|n-1} \forall i\in{1,\ldots,D}$, and thus the censoring rule in (8) becomes \emph{non-adaptive} across measurements at slot $n$.
	For AC-KF with \emph{non-adaptive} censoring, scalar entries of the measurement vector $\mathbf{y}_n$ can also be censored in a decentralized fashion across distributed sensors. }

\begin{algorithm}[t!]
	\caption{AC sketching module}\label{ACensor}
	\begin{algorithmic}
		\State {Measurement selection with AC-LMS}
		\State Input: $\hat{\boldsymbol{\theta}}_{n|n-1},~~\{\mathbf{y}_n,\mathbf{X}_n,\mathbf{R}_n\}$
		\State \underline{\textbf{Initialization:}} $\hat{\boldsymbol{\theta}}_{n|n-1,0}=\hat{\boldsymbol{\theta}}_{n|n-1},~\mathcal{S}_{n,0}=\emptyset$
		\For {$i=1:D$}
		\State Obtain $c_i$ as in \eqref{cis}
		\If {$c_i=0,$}
		\State $\mathcal{S}_{n,i}=\mathcal{S}_{n,i-1}\cup{\{i\}}$
		\State Update $\hat{\boldsymbol{\theta}}_{n|n-1,i-1}$ as in \eqref{update}
		\EndIf 
		\EndFor
		\vspace{-0.3cm}
		\begin{align*}  \mathrm{Return:}\{\check{\mathbf{y}}_n,\check{\mathbf{X}}_n,\check{\mathbf{R}}_n\}=\{\mathbf{y}_{n,{\mathcal{S}_{n,D}}},\mathbf{X}_{n,{\mathcal{S}_{n,D}}},\mathbf{R}_{n,{\mathcal{S}_{n,D}}}\}
		\end{align*}
	\end{algorithmic}
\vspace{-0.3cm}
\end{algorithm}

  \begin{table}[t!]\label{tab:table1}
  	\vspace{0.3cm}
  	\caption{{  Per time-slot sketching module complexity for $d<D$ data.}}
  	\label{tab:table_1}
  	\begin{center}
  		\begin{tabular} {|c|c|}
  			\hline
  			Sketching method & Complexity\\ \hline
  			Random sampling & $\mathcal{O}(D)$\\ \hline
  			Random projections (RP-KF) & $\mathcal{O}(Dp\log d )$\\ \hline
  			Adaptive censoring (AC-KF)& $\mathcal{O}(Dp)$\\ \hline
  		\end{tabular}
  	\end{center}
  \end{table}

\section{Update-Selection KF}
\label{sec:rc_KF}
In the last two sections, dimensionality reduction schemes were proposed for KF, by reducing the number of observations processed. {The resulting algorithms can be used to reduce the complexity of filtering as well as other data-related costs such as storage and transmission by reducing the dimensionality of $\mathcal{I}_n$. Specifically, if the observations need to be transmitted from a remote location, one could be interested in reducing the communication overhead as well computations; the RP-KF (where no feedback is required between the filter and the sensors)  or the AC-KF (that requires $\hat{\boldsymbol{\theta}}_{n|n-1}$ as feedback) would be preferable in such cases, with their dimensionality-reduction modules reducing the amount of data that need to be transmitted. On the other hand, if $D$ observations become available to the computing platform (block-by block or entry-by-entry) at each time-slot, reduction of computational complexity is the main concern, as well as the focus of the present section.  }  

Suppose that the observation noise has diagonal covariance matrix with $[\mathbf{R}_n]_{ii}=\sigma^2_i~\forall n$. Then, the KF correction step at time $n$ can be obtained by solving
\begin{align}\nonumber
\hat{\boldsymbol{\theta}}_{n|n}=\arg\min_{\boldsymbol{\theta}}&~\|\boldsymbol{\theta}-\hat{\boldsymbol{\theta}}_{n|n-1}\|_{\mathbf{P}_{n|n-1}^{-1}}^2\\\label{costf_seq}
&+\sum_{i=1}^D\frac{1}{\sigma^2_i}\left(y_{n,i}-\mathbf{x}_{n,i}^T\boldsymbol{\theta}\right)^2.
\end{align}
 A sequential (across entries of $\mathbf{y}_n$) solution of \eqref{costf_seq} can also be obtained using the recursive least-squares (RLS) algorithm with parameter estimate and error covariance matrix initialized at $\hat{\boldsymbol{\theta}}_{n|n-1}$ and $\mathbf{P}_{n|n-1}$, respectively. Specifically, let $\hat{\boldsymbol{\theta}}_{n|n,i}:=\mathbb{E}[\boldsymbol{\theta_n}|\hat{\boldsymbol{\theta}}_{n|n-1},\mathbf{y}_{n,{1:i}}]$ and $\mathbf{P}_{n|n,i}:=\mathrm{cov}(\hat{\boldsymbol{\theta}}_{n|n,i})$, for $i \in \{0,\ldots,D\}$; clearly, $\hat{\boldsymbol{\theta}}_{n|n,0}=\hat{\boldsymbol{\theta}}_{n|n-1}$ and $\hat{\boldsymbol{\theta}}_{n|n,D}=\hat{\boldsymbol{\theta}}_{n|n}$. The following RLS-like iteration, corresponding to the $i-$th entry of $\mathbf{y}_n$, updates the state estimate as
 \begin{align}\label{par_update}
\hat{\boldsymbol{\theta}}_{n|n,i}=\hat{\boldsymbol{\theta}}_{n|n,i-1}+ \mathbf{k}_{n,i}e_{n,i}
 \end{align} 
 where $e_{n,i}:=y_{n,i}-\mathbf{x}_{n,i}^T\hat{\boldsymbol{\theta}}_{n|n,i-1}$ and
 \begin{align}\label{gain}
 \mathbf{k}_{n,i}=\mathbf{P}_{n|n,i-1}\mathbf{x}_{n,i}{s}_{n,i}^{-1}
 \end{align} 
 with
  \begin{align}\label{innov}
  ~~~~~~~{s}_{n,i}:=\mathbf{x}_{n,i}^T\mathbf{P}_{n|n,i-1}\mathbf{x}_{n,i} + \sigma_i^2.
  \end{align} 
  The state covariance matrix is then updated as
 \begin{align}\label{covar_update}
 \mathbf{P}_{n|n,i}=\mathbf{P}_{n|n,i-1}-\mathbf{P}_{n|n,i-1}\mathbf{x}_{n,i}\mathbf{x}_{n,i}^T\mathbf{P}_{n|n,i-1}{s}_{n,i}^{-1}
 \end{align} 
  and the process is repeated until $i=D$, and all the measurements have been processed. 
  
  A common approach to dealing with $D\gg$ is to simply process as many data within time slot $n$ as the available computational resources allow for; see, e.g. \cite[Chapter 7]{bar2004estimation}. In the present work however, to reduce computational complexity, we propose \emph{judiciously skipping correction updates}. The criterion according to which the $i-$th row of $\{\mathbf{y}_n,\mathbf{X}_n\}$ will be used to update $\hat{\boldsymbol{\theta}}_{n|n,i-1}$ is based on how much the distribution $p({\boldsymbol{\theta}}_{n}|\hat{\boldsymbol{\theta}}_{n|n-1},[\mathbf{y}_n]_{1:i})$ after the update will diverge from the posterior $p({\boldsymbol{\theta}}_{n}|\hat{\boldsymbol{\theta}}_{n|n-1},[\mathbf{y}_n]_{1:i-1})$ before the update. A commonly used measure of difference between probability density functions (pdfs) is the Kullback–Leibler (KL) divergence, also known as relative entropy (see, e.g.~\cite{kullback1951information}). The KL divergence between two pdfs $p(x)$ and $q(x)$ is defined as{
  \begin{align}\nonumber
 \mathcal{D}_{KL}(p||q):=\int p(x) \ln \frac{p(x)}{q(x)} d x =\mathbb{E}_{p}\left[\ln \frac{p(x)}{q(x)}\right]
  \end{align}}
  and is not symmetric with respect to its arguments. In fact, one can interpret $p(x)$ as being the ``true'' pdf of $x$ while $q(x)$ is an approximate one. Then, $\mathcal{D}_{KL}(p||q)$ is a measure of how far the approximation is from reality.
  
  Aiming at carrying out only ``useful'' updates, $\mathcal{D}_{KL}(p_{n,i}|| p_{n,i-1})$ can be used as an indicator of how informative the update that involves the $i-$th row of $\{\mathbf{y}_n,\mathbf{X}_n\}$ is, where
    \begin{align}\nonumber
    p_{n,i}({\boldsymbol{\theta}}_{n}):=p({\boldsymbol{\theta}}_{n}|\hat{\boldsymbol{\theta}}_{n|n-1},[\mathbf{y}_n]_{1:i})
    \end{align}
  and
      \begin{align}\nonumber
      p_{n,i-1}({\boldsymbol{\theta}}_{n}):=p({\boldsymbol{\theta}}_{n}|\hat{\boldsymbol{\theta}}_{n|n-1},[\mathbf{y}_n]_{1:i-1}).
      \end{align}
  
   \begin{proposition}\label{pro:dkl}
   	Let observations be generated according to \eqref{model2a}-\eqref{model2b} with ${\mathbf{w}_n}$ and ${\mathbf{v}_n}$ Gaussian. Then, for the sequential estimator in \eqref{par_update}-\eqref{covar_update} it holds that 
   	    \begin{align*}\nonumber
         \mathcal{D}_{KL}(p_{n,i}|| p_{n,i-1})=&\frac{1}{2}\left(\bar{e}_{n,i}^2-1\right)\frac{\gamma_{n,i}}{\gamma_{n,i}+\sigma_i^2}\\
         & + \ln\sqrt{\frac{\gamma_{n,i}+\sigma_i^2}{\sigma_i^2}}
   	    \end{align*}
   	    where $\gamma_{n,i}:=\mathbf{x}_{n,i}^T\mathbf{P}_{n|n,i-1}\mathbf{x}_{n,i}$, and $\bar{e}_{n,i}:=e_{n,i}s_{n,i}^{-1/2}$ is the per-entry normalized innovation.
  \end{proposition}
 
  \begin{proof}
  	See Appendix 3.
  \end{proof}
  Proposition \ref{pro:dkl} offers a simple expression of $\mathcal{D}_{KL}(p_{n,i}|| p_{n,i-1})$ that will come handy in performing informative updates. Consider first the quantities involved in $\mathcal{D}_{KL}(p_{n,i}|| p_{n,i-1})$, namely the normalized innovation (residual) $\bar{e}_{n,i}$, which is a random variable and $\gamma_{n,i}$ that is deterministic. Interestingly, $\gamma_{n,i}:=\|\mathbf{x}_{n,i}\|_{\mathbf{P}_{n|n,i}}^2$ is an \emph{ensemble} quantity capturing the expected power of the $i-$th observation across the main directions of state uncertainty, while $\bar{e}_{n,i}$ is a random data-dependent variable that measures how important the $i-$th update is for a specific realization of the problem. Depicted in Fig.~\ref{fig:dkl} is a simulated sequence of $\mathcal{D}_{KL}(p_{n,i}|| p_{n,i-1})$ as a function of index $i$ for an arbitrary time-slot $n$. Immediately noticeable is that the per-step divergence decreases with an approximate rate of $1/i$ following the rate of decrease of $\mathbf{P}_{n|n,i}$. 
  One may also observe that certain updates yield higher KL divergence compared to the baseline. Fast and computationally efficient estimation may be achieved by isolating and performing only such informative updates. 
  
   While $\mathcal{D}_{KL}(p_{n,i}|| p_{n,i-1})$ offers a measure of difference between the posteriors prior and after updating, it lacks symmetry and it is not conveniently interpreted as distance. Consequently, we considered the modified metric 
    \begin{align}\label{sym_dkl}
    \mathcal{D}(p_{n,i}, p_{n,i-1}):=\mathcal{D}_{KL}(p_{n,i}|| p_{n,i-1})+\mathcal{D}_{KL}(p_{n,i}|| p_{n,i-1})
    \end{align} 
  also known as the symmetric KL divergence. As seen in Figs.~\ref{fig:dkl} and \ref{fig:sym_dkl}, both metrics follow a similar trend and converge to $0$ as the state estimate converges in probability. Nevertheless, $\mathcal{D}(p_{n,i}, p_{n,i-1})$ enjoys symmetry as well as a more simple expression which as given in Proposition \ref{pro:sym_dkl}. Subsequently, the following rule is proposed for selecting informative updates
  \begin{align}\label{update_rule}
  \mathcal{D}(p_{n,i}|| p_{n,i-1})\bigg\{
  \begin{array}{lc}
  \geq \tau_{n,i},& \mathrm{Update}~~\hat{\boldsymbol{\theta}}_{n|n,i},\mathbf{P}_{n|n,i} \\
  < \tau_{n,i},& \mathrm{Skip ~updates}.
  \end{array}
  \end{align} 
  
  \begin{proposition}\label{pro:sym_dkl}
  	Let observations be generated according to \eqref{model2a}-\eqref{model2b} with ${\mathbf{w}_n}$ and ${\mathbf{v}_n}$ Gaussian. Then, for the sequential estimator in \eqref{par_update}-\eqref{covar_update} it holds that 
  	\begin{align}\label{exact_sym_dkl}
  	\mathcal{D}(p_{n,i}, p_{n,i-1})=\frac{1}{2}\bar{e}_{n,i}^2\left(2\gamma_{n,i}+\left(\frac{\gamma_{n,i}}{\sigma_i}\right)^2\right)s_{n,i}^{-1}.
  	\end{align}
  \end{proposition}
    \begin{proof}
    	See Appendix 4.
    \end{proof}
  Using \eqref{exact_sym_dkl}, rule \eqref{update_rule} can be readily implemented. Regarding the sequence of thresholds $\{\tau_{n,i}\}_{i=1}^D$, a judicious choice is  
  \begin{align}\label{thresholds}
  \tau_{n,i}=\tau_n \frac{1}{i}
  \end{align}
  which promotes updates with large informational value relative to the stage of the estimation process. The total number of updates per slot $n$ can be tuned by $\tau_n$.	A couple of remarks are now in order. 
  	 
  \begin{figure}[t!]
  	\centering
  	\subfigure[]{
  		\centerline { \includegraphics[width=0.5\textwidth]{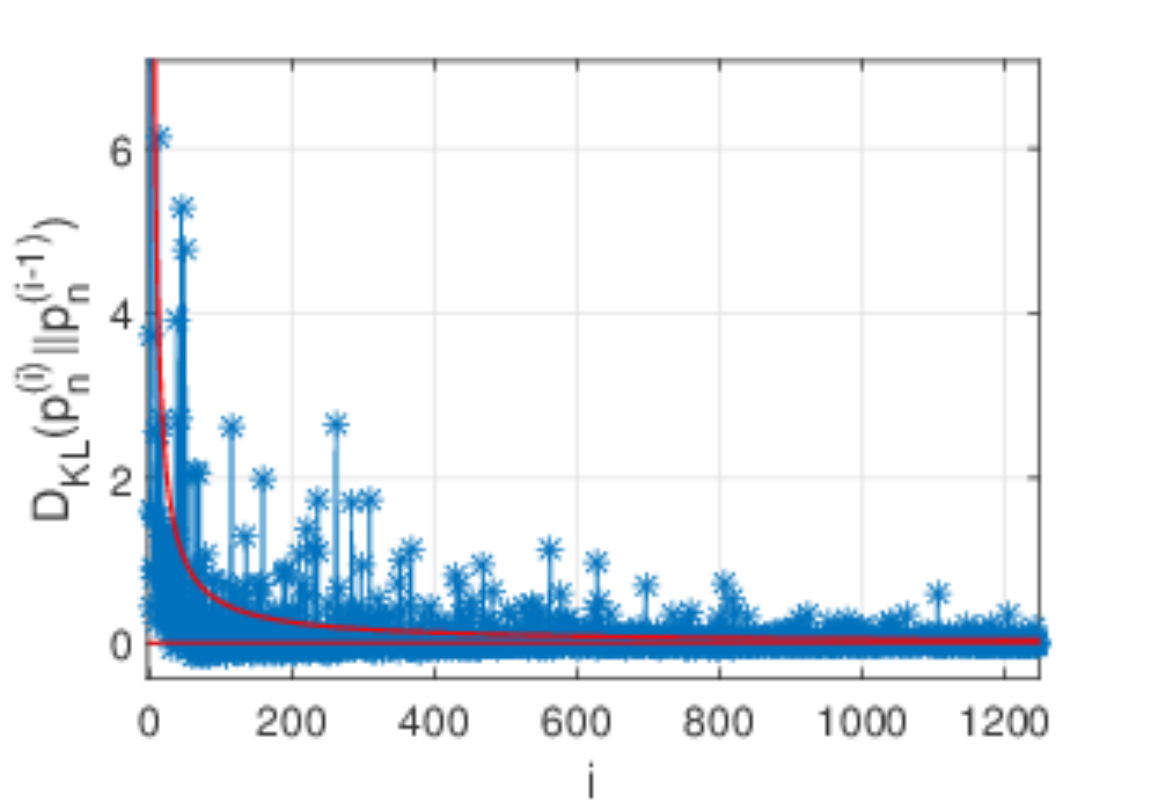}}	
  		\label{fig:dkl}
  	}
  	\subfigure[]{
  	   \centerline	{\includegraphics[width=0.5\textwidth]{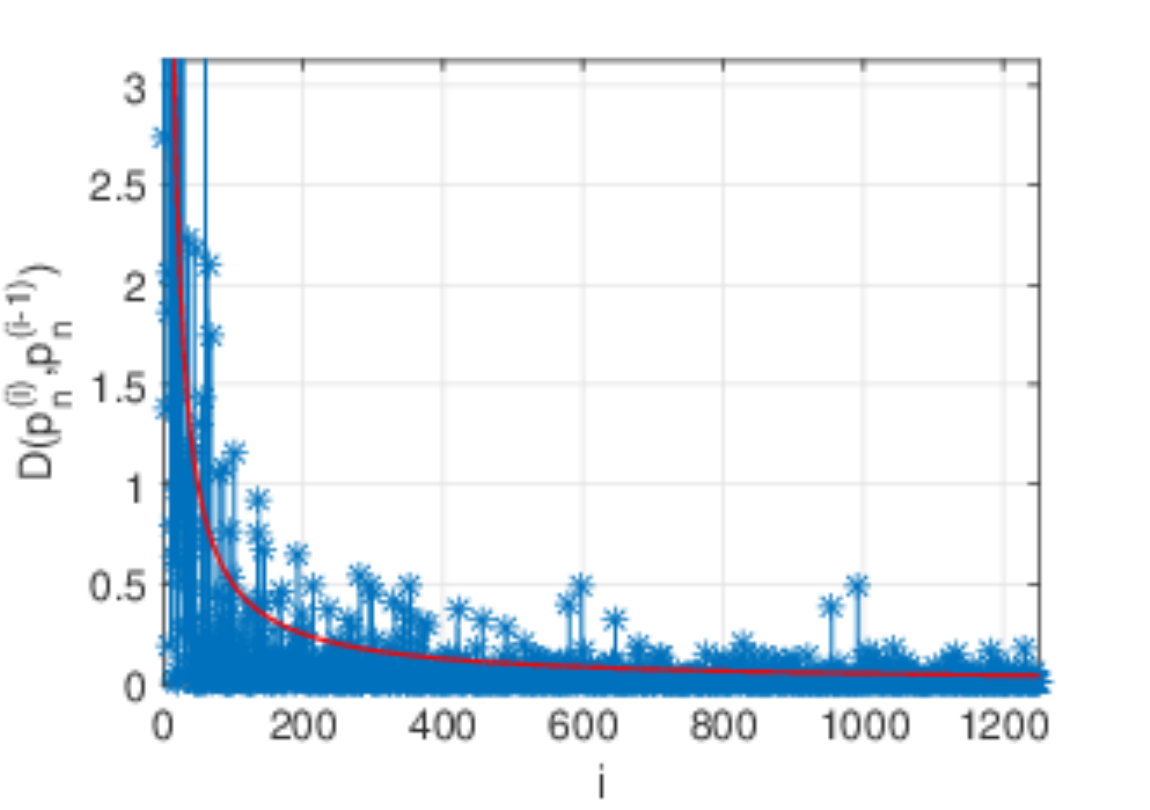}}
  		\label{fig:sym_dkl}
  	}
  	\caption{Example of a) KL divergence and b) symmetric KL divergence  evolution across sequential correction updates. Both metrics converge to $0$ following a $\propto 1/i$ trend.}
  \end{figure}

  \textbf{Remark 2:} KL divergence induced by a measurement was also employed by \cite{zheng2014sequential} to offer an alternative viewpoint on a \emph{distributed} censoring rule for reducing the \emph{communication load} in WSNs. Specifically, it was shown that the KL divergence of $p({\boldsymbol{\theta}_n})$ with $p({\boldsymbol{\theta}_n}|y_{n,i}))$ as reference $\big($i.e. $\mathcal{D}_{KL}(p({\boldsymbol{\theta}_n})|| p({\boldsymbol{\theta}_n}|y_{n,i}))\big)$ is proportional to the magnitude of $e_{n,i}$, which implies that the latter is related to the informational value of a measurement. Apart from the different goals and context, a major difference of the present section's contribution relative to~\cite{zheng2014sequential} is the explicit calculation, and use of the (symmetric) KL divergence in the proposed update selection rule.
  
   \textbf{Remark 3:} Interestingly, our proposed \emph{data-driven} update selection using $\mathcal{D}_{KL}(p_{n,i}|| p_{n,i-1})$ is also related to OED-type sensor selection schemes that are based on the mutual information between a sensor and the model (e.g.,~\cite{ertin2003maximum}). This relation can be observed upon recalling that the mutual information $I(X;Y)$ between two random variables $X$ and $Y$ can be expressed as
   \begin{equation}\label{mutual}
   I(X;Y)=\mathbb{E}_Y[\mathcal{D}_{KL}(p(X|Y)|| p(X))].
   \end{equation}
   In the present context,~\eqref{mutual} implies that the mutual information between the $i-$th ``sensor'' at time slot $n$ and $\boldsymbol{\theta}_n$ in a sequential processing setting, equals its $\mathcal{D}_{KL}(p_{n,i}|| p_{n,i-1})$ averaged over all possible measurements $y_{n,i}$.

  \subsection{Reduced-complexity censoring rule}
  
  In the previous section, an update selection rule was introduced in \eqref{update_rule} relying on the information metric in \eqref{exact_sym_dkl}. Practical implementation of \eqref{update_rule} requires careful consideration of the computational complexity needed to obtain $\mathcal{D}(p_{n,i}, p_{n,i-1})$. As seen in \eqref{exact_sym_dkl}, to obtain the latter it suffices to compute $e_{n,i}$ and $\gamma_{n,i}$ (since $s_{n,i}=\gamma_{n,i}+\sigma_i^2$). While computing $e_{n,i}$ requires only $\mathcal{O}(p)$ products, obtaining $\gamma_{n,i}:=\mathbf{x}_{n,i}^T\mathbf{P}_{n|n,i-1}\mathbf{x}_{n,i}$ requires a matrix-vector product that comes with $\mathcal{O}(p^2)$ complexity. Thus, even though checking whether an update is informative or not has smaller complexity than the update itself (cf. \eqref{par_update}-\eqref{covar_update}), both tasks are of the same order of $\mathcal{O}(p^2)$ complexity. Ideally, checking the update should be less costly than performing the update by an order of magnitude.
  
  For this purpose, a low-complexity approximation of $\gamma_{n,i}$ is highly desirable. One way to approximate $\gamma_{n,i}$ is to use the eigen-decomposition $\mathbf{P}_{n|n,i-1}=\mathbf{V}_{n,i-1}\mathbf{\Lambda}_{n,i-1}\mathbf{V}_{n,i-1}^T$ to produce the best $k$-rank approximation of $\mathbf{P}_{n|n,i-1}$ as
    \begin{align}\nonumber
    \hat{\mathbf{P}}_{n|n,i-1}^{k}=\mathbf{V}_{n,i-1}^{k}\mathbf{\Lambda}_{n,i-1}^{k}(\mathbf{V}_{n,i-1}^{k})^T
    \end{align}
   where $\mathbf{\Lambda}_{n,i-1}^{k}$ is a $k\times k$ diagonal matrix containing the $k\leq p$ largest eigenvalues of $\mathbf{P}_{n|n,i-1}$, and $\mathbf{V}_{n,i-1}^{k}$ is the $p\times k$ matrix of the corresponding eigenvectors. Using $\hat{\mathbf{P}}_{n|n,i-1}^{k}$ to approximate $\gamma_{n,i}$, yields
       \begin{align}\nonumber
       g_{n,i}^{k}&:=\mathbf{x}_{n,i}^T\hat{\mathbf{P}}_{n|n,i-1}^{k}\mathbf{x}_{n,i}\\\label{naive}
       &=\mathbf{x}_{n,i}^T\mathbf{V}_{n,i-1}^{k}\mathbf{\Lambda}_{n,i-1}^{k}(\mathbf{x}_{n,i}^T\mathbf{V}_{n,i-1}^{k})^T
       \end{align}
    which can be obtained with $\mathcal{O}(pk)$ complexity. Although $g_{n,i}^{k}$ exactly captures the power of $\mathbf{x}_{n,i}^T$ along the principal directions of $\mathbf{P}_{n|n,i-1}$, it is in general a poor estimate of $\gamma_{n,i}$. In fact, ignoring the power of the $p-k$ smallest eigenvalues of $\mathbf{P}_{n|n,i-1}$ leads to under-estimation; that is $g_{n,i}^{k}\leq{\gamma}_{n,i}^{k}$. To mitigate this problem, let us denote the power of $\mathbf{x}_{n,i}$ that is not captured by the first $k$ principal directions as 
           \begin{align}\nonumber
           g_{n,i}^{-k}:=\|\mathbf{x}_{n,i}\|_2^2-\|\mathbf{x}_{n,i}^T\mathbf{V}_{n,i-1}^{k}\|_2^2.
           \end{align}
  Then, if the the remaining power $g_{n,i}^{-k}$ is distributed evenly along the $p-k$ directions that correspond to the smallest eigenvalues of $\mathbf{P}_{n|n,i-1}$, an improved approximation of ${\gamma}_{n,i}^{k}$ is 
             \begin{align}\nonumber
             \hat{\gamma}_{n,i}^{k}&=g_{n,i}^{k}+\frac{1}{p-k}g_{n,i}^{-k}\sum_{j=k+1}^p[\mathbf{\Lambda}_{n,i-1}]_{jj}\\\label{approx_g}
             &=g_{n,i}^{k}+\frac{1}{p-k}g_{n,i}^{-k}\left[\mathrm{tr}(\mathbf{P}_{n|n,i-1})-\mathrm{tr}(\mathbf{\Lambda}_{n,i-1}^{k})\right]
             \end{align}
  where for the second equality we used that $\mathrm{tr}(\mathbf{P}_{n|n,i-1})=\mathrm{tr}(\mathbf{\Lambda}_{n|n,i-1})$. Essentially, by ignoring the angle of $\mathbf{x}_{n,i}^T$ along the $p-k$ least important directions of uncertainty, an estimate $\hat{\gamma}_{n,i}^{k}\approx{\gamma}_{n,i}$ can be found with $\mathcal{O}(pk)$ complexity. Simulations will demonstrate that for most cases $k$ need not be very large for the purpose of obtaining a reliable approximation of ${\gamma}_{n,}^{(i)}$, and thus of the update selection rule in \eqref{update_rule}. 
  
  Finally, let us consider the computational burden of re-computing the eigen-decomposition of $\mathbf{P}_{n|n,i-1}$ when an update is performed. Fortunately, the decomposition needs only be fully computed once for $\mathbf{P}_{n|n,0}$, after the prediction step. Then, exploiting that $\mathbf{P}_{n|n,i-1}$ is given as a sequence of rank-one updates of symmetric positive matrices (cf. \eqref{covar_update}) allows for low-complexity $\mathcal{O}(p^2)$ updates of the eigen-decompositions, see e.g.~\cite{yu1991recursive} and~\cite{gu1994stable}, while the fact that only the first $k$ eigen-pairs are required can further reduce the complexity of the updates. The need for tracking the principal eigen-pairs of  $\mathbf{P}_{n|n,i-1}$ can be completely eliminated by setting $k=0$, which yields the estimate
               \begin{align}\label{bit_naive}
                 \hat{\gamma}_{n,i}^{0}=\frac{1}{p}\|\mathbf{x}_{n,i}\|_2^2\mathrm{tr}(\mathbf{P}_{n|n,i-1})
                 \end{align}
  with $\mathcal{O}(p)$ complexity, at the cost of ignoring information given by the angle of $\mathbf{x}_{n,i}$. For cases where the eigenvalues of $\mathbf{P}_{n|n,i-1}$ are approximately uniform, \eqref{bit_naive} provides a practical and sufficiently accurate estimate of ${\gamma}_{n,i}$. Generally, obtaining \eqref{approx_g} requires $\mathcal{O}(p(k+1))$ computations. Overall, the complexity of the correction step using the iterative method in \eqref{par_update}-\eqref{covar_update} with the update selection rule in \eqref{update_rule} and the approximation in \eqref{approx_g} is $\mathcal{O}(dp^2)+\mathcal{O}(Dp(k+1))$, where $d\ll D$ is the number of updates. Depending on the size of $p,d$ and $D$, the overall complexity of the proposed scheme can be considerably less than the standard $\mathcal{O}(Dp^2)$.
  
  \subsection{First-order updates}
  
  The fact that the update selection rule in \eqref{update_rule} requires at least $\mathcal{O}(p)$ computations hints at possible modifications of the present scheme, that are considered in this section. Specifically, instead of using \eqref{update_rule} to completely skip updates, one may incorporate a simple $\mathcal{O}(p)$ update without noticeably increasing the overall complexity of the algorithm. For instance, having computed $e_{n,i}$ which is required in \eqref{update_rule}, the following LMS-like parameter update can readily be implemented 
  \begin{equation}\label{1sto_update} 
  \hat{\boldsymbol{\theta}}_{n|n,i}=\hat{\boldsymbol{\theta}}_{n|n,i-1}
  +\mu_{n,i} \mathbf{x}_{n,i} e_{n,i}
  \end{equation}
  where $\mu_{n,i}$ denotes a user selected stepsize. Given the update in \eqref{1sto_update}, consider the difference
   \begin{equation}\label{delta_mse} 
   \Delta_n(\mu_{n,i}):=\mathrm{MSE}(\hat{\boldsymbol{\theta}}_{n|n,i})-\mathrm{MSE}(\hat{\boldsymbol{\theta}}_{n|n,i-1})
   \end{equation}
  where $\mathrm{MSE}(\boldsymbol{\theta}):=\mathbb{E}[\|\boldsymbol{\theta}-\boldsymbol{\theta}_n\|_2^2]$. Clearly, $\mu_{n,i}$ should be chosen such that $\Delta_n(\mu_{n,i})\leq{0}$, or, ideally such that $\Delta_n(\mu_{n,i})$ is minimized. But first, it is useful to derive an explicit expression for $\Delta_n(\mu_{n,i})$.   

\begin{proposition}\label{pro:delta_mse}
    For an update of the form \eqref{1sto_update} that follows an update of \eqref{par_update}-\eqref{covar_update} it holds that 
	\begin{align}\label{exact_delta_mse}
	\Delta_n(\mu_{n,i})=  \|\mathbf{x}_{n,i}\|_2^2 (\gamma_{n,i} +\sigma_i^2)\mu_{n,i}^2 -2\gamma_{n,i}\mu_{n,i}.
	\end{align}
\end{proposition}
\begin{proof}
	See Appendix 5.
\end{proof}
  To guarantee $\Delta_n(\mu_{n,i})\leq{0}$, it suffices to choose $\mu_{n,i}$ as
  	\begin{align}\label{suff}
0 \leq \mu_{n,i} \leq \frac{2\gamma_{n,i}}{\|\mathbf{x}_{n,i}\|_2^2(\gamma_{n,i} +\sigma_i^2)}
  	\end{align}
  while for
    	\begin{align}\label{opt}
    	\mu_{n,i}^{\ast} = \frac{\gamma_{n,i}}{\|\mathbf{x}_{n,i}\|_2^2(\gamma_{n,i} +\sigma_i^2)}
    	\end{align}
  the minimum of $\Delta_n(\cdot)$ is attained
  	\begin{align}\label{min}
  	\Delta_n(\mu_{n,i}^{\ast})= -\frac{(\gamma_{n,i})^2}{\|\mathbf{x}_{n,i}\|_2^2(\gamma_{n,i} +\sigma_i^2)}. 
  	\end{align}
 Although the LMS-like iteration \eqref{1sto_update} reduces the MSE by as much as $-\Delta_n(\mu_{n,i}^{\ast})$, updating $\mathbf{P}_{n|n,i-1}$ incurs complexity $\mathcal{O}(p^2)$, and it is thus skipped. Skipping covariance updates for first-order updates is also well motivated by the fact that $-\Delta_n(\mu_{n,i}^{\ast})$ is generally significantly smaller than the reduction achieved by the second-order updates \eqref{par_update}-\eqref{covar_update}. Nevertheless, one may in practice use  $\mu_{n,i}\leq\mu_{n,i}^{\ast}$, to compensate for the (slow) decrease in estimation variance. Finally, while the exact value of $\gamma_{n,i}$ is generally not available, $\mu_{n,i}$ can be selected after using the estimate $\hat{\gamma}_{n,i}$ from \eqref{approx_g} in \eqref{suff} or \eqref{opt}. The overall proposed reduced-complexity update-selection (US) KF described in Section \ref{sec:rc_KF} is tabulated as Algorithm \ref{rc_KF}.   
   	\begin{algorithm}[t!]
   		\caption{US-KF}\label{rc_KF}
   		\begin{algorithmic}
   			\State \underline{\textbf{Initialization:}} $\hat{\boldsymbol{\theta}}_{0|0}=\mathbf{m}_0,~\mathbf{P}_{0|0}=\mathbf{P}_0$
   			\For {$n=1:N$}
   			\State \underline{\textbf{Prediction Step:}}
   			\State \hspace{0.07 cm}$\hat{\boldsymbol{\theta}}_{n|n-1}=\mathbf{F}_n\hat{\boldsymbol{\theta}}_{n-1|n-1}+\mathbf{G}_n\mathbf{u}_n$
   			\State $\mathbf{P}_{n|n-1}=\mathbf{F}_n\mathbf{P}_{n-1|n-1}\mathbf{F}_n^T+\mathbf{Q}_{n}$
   			\State \underline{\textbf{Correction Step:}}
   			\State Set parameters: $k,\tau_n$
   			\State  Initialize: $\hat{\boldsymbol{\theta}}_{n|n,0}=\hat{\boldsymbol{\theta}}_{n|n-1}, \mathbf{P}_{n|n,0}=\mathbf{P}_{n|n-1}$
   			\State Compute $k$ first eigenpairs:  $\{\mathbf{V}_{n,0}^{k},\mathbf{\Lambda}_{n,0}^{k}\}$
   			\For {$i=1:D$}
   			\State Obtain $\hat{\gamma}_{n,i}$ as in \eqref{approx_g}
   			\State Obtain $\mathcal{D}(p_{n,i}, p_{n,i-1})$ as in \eqref{exact_sym_dkl}
   			\If {$\mathcal{D}(p_{n,i}, p_{n,i-1})\geq \tau_n/i$} 
   			\State Update $\{\hat{\boldsymbol{\theta}}_{n|n,i}, \mathbf{P}_{n|n,i}\}$ as in \eqref{par_update}-\eqref{covar_update}
   			\State Update $\{\mathbf{V}_{n,i}^{k},\mathbf{\Lambda}_{n,i}^{k}\}$ of rank-$1$ update of $\mathbf{P}_{n|n,i}$
   			\Else
   			\State Obtain $\mu_{n,i}$ by plugging $\hat{\gamma}_{n,i}$ in \eqref{suff} or \eqref{opt}
   			\State Update $\hat{\boldsymbol{\theta}}_{n|n,i}$ as in \eqref{1sto_update}
   			\EndIf
   			\EndFor
   			\EndFor
   		\end{algorithmic}
   	\end{algorithm}
%

\section{Budgeted Fixed-Interval Smoothing}

The methods introduced in Sections \ref{sec:rpkf},~\ref{sec:AC-KF} and \ref{sec:rc_KF} utilize dimensionality reduction, measurement selection, and update selection, in order to promote low-complexity correction updates of the KF. In the present section, we briefly explore another direction that allows for reliable tracking with smaller data usage and computational complexity. Specifically, consider the ``smoothed'' estimate $\hat{\boldsymbol{\theta}}^{\mathrm{KS}}_n:=\mathbb{E}[\boldsymbol{\theta}_n|\{\mathbf{y}_n\}_{n=1}^N]$, { and let $\hat{\boldsymbol{\theta}}^{\mathrm{KS}}$ be formed by concatenating all such smoothed estimates. This can also be written as
				\begin{align}\nonumber
					\hat{\boldsymbol{\theta}}^{\mathrm{KS}}&=\arg\min_   {\{\boldsymbol{\theta}_n\}_{n=1}^N}\frac{1}{2}\sum\limits_{n=1}^{N}\|\mathbf{y}_n-\mathbf{X}_n\boldsymbol{\theta}_n\|_{\mathbf{R}_n^{-1}}^2\\
					&+\|\boldsymbol{\theta}_n-\mathbf{F}_n\boldsymbol{\theta}_{n-1}\|_{\mathbf{Q}_n^{-1}}^2+\|\boldsymbol{\theta}_0-\mathbf{m}_0\|_{\mathbf{P}_n^{-1}}^2\label{KS}
				\end{align}}
				which is optimal in the linear minimum mean-square error (LMMSE) sense. 
				
				Aiming at a recursive solver of~\eqref{KS}, one can rely on the Rauch-Tung-Stribel (RTS) forward-backward KS algorithm~\cite{rts}. In its forward pass, the RTS algorithm is identical to the KF.
				The KF estimates $\{\hat{\boldsymbol{\theta}}_{n|n}\}_{n=1}^N$ are then stored and processed by the backward pass of the KS, while the error covariance matrices $\{\mathbf{P}_{n|n}\}_{n=1}^N$ are computed off-line.
				
				Given $\hat{\boldsymbol{\theta}}_{n+1|N}$, the backward iteration solves
				\begin{align}\label{costs}
					\hat{\boldsymbol{\theta}}_{n|N}:=\arg\min_{\boldsymbol{\theta}}\|\hat{\boldsymbol{\theta}}_{n+1|N}-\mathbf{F}_{n}\boldsymbol{\theta}\|_{\mathbf{Q}_n^{-1}}^2+\|\boldsymbol{\theta}-\hat{\boldsymbol{\theta}}_{n|n}\|_{\mathbf{P}_{n|n}^{-1}}^2.
				\end{align}
				Similar to filtering, the minimizer of~\eqref{costs} is also given in closed form as
				\begin{align*}
					\hat{\boldsymbol{\theta}}_{n|N}=\left(\mathbf{F}_n^T\mathbf{Q}_n^{-1}\mathbf{F}_n+\mathbf{P}_{n|n}^{-1}\right)^{-1} 
					\left(\mathbf{F}_n^T\mathbf{Q}_n^{-1}\hat{\boldsymbol{\theta}}_{n+1|N}+\mathbf{P}_{n|n}^{-1}\hat{\boldsymbol{\theta}}_{n|n}\right).
				\end{align*}
				After invoking the MIL and letting $\mathbf{B}_n:=\mathbf{P}_{n|n}\mathbf{F}_n^T\mathbf{P}_{n+1|n}^{-1}$, the estimate $\hat{\boldsymbol{\theta}}_{n|N}$ is given in the correction form of $\hat{\boldsymbol{\theta}}_{n|n}$ as  
					\begin{align}
						\hat{\boldsymbol{\theta}}_{n|N}=\hat{\boldsymbol{\theta}}_{n|n}+\mathbf{B}_n\left(\hat{\boldsymbol{\theta}}_{n+1|N}-\mathbf{F}_n\hat{\boldsymbol{\theta}}_{n|n}\right)
					\end{align}
				with corresponding error covariance matrix
				\begin{align}\label{KSerror}
					~~~~~~~~~~\mathbf{P}_{n|N}=\mathbf{P}_{n|n}+\mathbf{B}_n\left(\mathbf{P}_{n+1|N}-\mathbf{P}_{n+1|n}\right)\mathbf{B}_n^T.
				\end{align} 
				A key property of the backward KS iteration, is that it improves KF performance using from $\{\mathcal{I}_n\}_{n=1}^N$ only the information encapsulated in the output $\hat{\boldsymbol{\theta}}_{n|n}$ of the forward filter. Therefore, backward iterations can be readily applied on filtered estimates of RP-KF, AC-KF or the US-KF to limit the tracker's performance loss caused by the measurement reduction.

				\begin{algorithm}[t!]
					\caption{The budgeted Kalman smoother (Bud-KS)}\label{algorithm2}
					\begin{algorithmic}
						\For {$n=N-1:0$}
						\If{$\hat{\boldsymbol{\theta}}_{n|n}\in\Theta^S_n$}
						\State $\hat{\boldsymbol{\theta}}_{n|N}=\hat{\boldsymbol{\theta}}_{n|n}$
						\State $\mathbf{P}_{n|N}=\mathbf{P}_{n|n}$
						\Else
						\State $\hat{\boldsymbol{\theta}}_{n|N}=\hat{\boldsymbol{\theta}}_{n|n}+\mathbf{B}_n\left(\hat{\boldsymbol{\theta}}_{n+1|N}-\mathbf{F}_n\hat{\boldsymbol{\theta}}_{n|n}\right)$
						\State $\mathbf{B}_n=\mathbf{P}_{n|n}\mathbf{F}_n^T\mathbf{P}_{n+1|n}^{-1}$
						\State $\mathbf{P}_{n|N}=\mathbf{P}_{n|n}+\mathbf{B}_n\left(\mathbf{P}_{n+1|N}-\mathbf{P}_{n+1|n}\right)\mathbf{B}_n^T$
						\EndIf
						\EndFor
					\end{algorithmic}
				\end{algorithm}
				
				In addition, the backward iteration can also be modified to operate within a limited computational budget. Given the smoothed estimate at time $n+1$, let us define the set
				\begin{align}\label{infoset}
					\Theta^b_n:=\left\{\boldsymbol{\theta}\big{|}\|\hat{\boldsymbol{\theta}}_{n+1|N}-\mathbf{F}_{n}\boldsymbol{\theta}\|_{\mathbf{Q}_n^{-1}}^2\leq{\tau_b}\right\}
				\end{align} 
				of states at time $n$ that are consistent enough with the transition model in the WLS sense. Based on~\eqref{infoset}, the Bud-KS estimate at time $n$ is given as 
				\begin{align}\label{budKS}
					\hat{\boldsymbol{\theta}}_{n|N}=\bigg\{
					\begin{array}{lc}
						\hat{\boldsymbol{\theta}}_{n|n},&{\hat{\boldsymbol{\theta}}_{n|n}\in\Theta^b_n}\\
						\hat{\boldsymbol{\theta}}_{n|n}+\mathbf{B}_n\left(\hat{\boldsymbol{\theta}}_{n+1|N}-\mathbf{F}_n\hat{\boldsymbol{\theta}}_{n|n}\right),&{\hat{\boldsymbol{\theta}}_{n|n}\notin\Theta^b_n}.
					\end{array}
				\end{align} 
				Clearly, for $\hat{\boldsymbol{\theta}}_{n|n}\in\Theta^b_n$, it holds that $\mathbf{P}_{n|N}=\mathbf{P}_{n|n}$; while for $\hat{\boldsymbol{\theta}}_{n|n}\notin\Theta^b_n$, the error covariance is given by~\eqref{KSerror}. Essentially, KS estimates that are consistent enough with the system model are not smoothed, thus saving the computations required. Here, the threshold $\tau_b$ in~\eqref{infoset} can be tuned to control the amount of ``acceptable'' deviation from the model. The novel economical, fixed-interval smoother on a budget, that we abbreviate as Bud-KS, is tabulated as Algorithm~\ref{algorithm2}.
				
				{Regarding the computational complexity of Bud-KS, it is worth noting that implementing the rule (33) in the general case requires $\mathcal{O}(p^3)$ computations in order to invert $\mathbf{Q}_n$. The complexity of Bud-KS updates in (31) and (32) are on the same order of magnitude. Thus, Bud-KS is preferable 
				when the covariance matrix of $\mathbf{w}_n$ is time-invariant,  meaning that $\mathbf{Q}_n=\mathbf{Q}\; \forall n$. In such cases, inversion of $\mathbf{Q}$ is  performed once offline, thus reducing complexity in (33) to $\mathcal{O}(p^2)$; likewise, when $\mathbf{Q}_n$ is diagonal. In such scenarios, an update of $\mathcal{O}(p^3)$ complexity is skipped at the cost of an $O(p^2)$ complexity rule, leading to computational savings that become more significant as $p$ increases.}

\section{Numerical Tests}\label{sec:simulations}

The novel AC-KF, RP-KF, US-KF and Bud-KS algorithms are tested here on a simulated linear dynamical system. For this experiment, a simple state transition model that performs cyclical shifting of the entries of the state was implemented. The state transition matrix is 
 \begin{align*}
 \mathbf{F}_{n,ij}\bigg\{
 \begin{array}{lc}
 1,& i=j-1 \\            
 0,& \mathrm{otherwise}
 \end{array},~~\forall{n}
 \end{align*} 
	and $\mathbf{F}_{1p}=1$, while the state dimension is set to $p=50$. The state noise $\{\mathbf{w}_n\}_{n=1}^N$ was generated i.i.d. with $\mathbf{w}_n\sim{\mathcal{N}(\mathbf{0},\sigma_{w}^2\mathbf{Q}_n)}$, where $\mathbf{Q}_{n,ij}=0.5^{|i-j|}$ and $\sigma_{w}=0.01$. Finally, the initial state is $\boldsymbol{\theta}_0\sim{\mathcal{N}(\mathbf{m}_0,\mathbf{P}_0)}$, with $\mathbf{m}_0$ set to have two non-zero values $20$ and $-30$ in its first and fifth entry, and $\mathbf{P}_0=0.04\mathbf{I}$.
	Per time instant $n\in\{1, \ldots, N\}$ with $N=100$, $D=500$ measurements are obtained and concatenated in $\mathbf{y}_n=\mathbf{X}_n\boldsymbol{\theta}_{n}+\mathbf{v}_n$, where rows of $\mathbf{X}_n$ are generated as i.i.d. standardized Gaussian vectors and then weighted independently by coefficients $\alpha$ drawn from $\alpha\sim{\textrm{Unif}\{0.5,~1.5 \}}$. For this experiment, observations are correlated; thus, $\mathbf{v}_n\sim{\mathcal{N}(0,\sigma_v^2\mathbf{R}_n)}$, where $\mathbf{R}_{n,ij}=0.5^{|i-j|}$. For the following experiments, we set $\sigma_v^2=1$, upon observing that the results remain qualitatively similar for different noise levels.  

	\subsection{AC-KF, RP-KF, and US-KF}
	
	To determine the average performance in terms of estimation error and computational complexity of AC-KF and RP-KF for different values of $d/D$, 20 Monte Carlo realizations were run on the same simulated linear dynamical system. The estimation performance was measured in terms of root mean-square error (RMSE) of the estimates across iterations; that is,{
	\begin{equation*}
		\mathrm{RMSE}=\sqrt{\frac{1}{N}\sum_{n=1}^{N}\|\hat{\boldsymbol{\theta}}_{n|n}-\boldsymbol{\theta}_{n}\|_2^2}.
	\end{equation*} }
	AC-KF and US-KF were run first, with thresholds tuned such that a constant number of approximately $d$ observations were selected per time slot; RP-KF and the greedy algorithm were then set to obtain $d$ measurements per time slot. As a performance benchmark for the three algorithms, KF was also run with $d$ randomly sampled observations per time step.
	
	The average RMSE of the five methods as a function of $d/D$ is plotted in Fig.~\ref{fig:hsnr}. These plots confirm that the proposed data-agnostic RP-KF is useful for increasing the accuracy (compared to plain random sampling) when estimating dynamic processes. With regards to the more elaborate algorithms, AC-KF has comparable performance with the KF using greedy OED measurement selection, while being orders of magnitude faster in terms of runtime. Last but not least, the US-KF with $k=0$ outperforms the other methods while maintaining $\mathcal{O}(dp^2)$ complexity, even when the observation noise is correlated. { Finally, the experiment was re-run with $D=1000$ and for varying $d/D$, with the runtime of the algorithms listed in Table \ref{tab:table_2}. The greedy ODE algorithm is excluded from this experiment since it is an offline benchmark with runtime larger than that of the full-data KF. In comparison to random sampling, the proposed methods carry a certain computational overhead which becomes less significant as $d/D$ (or $D$) increases. More importantly, the proposed algorithms enjoy a significantly lower runtime than the full-data KF.} 
	
	Additional experiments were performed to assess sensitivity of the US-KF to the choice of parameter $k$. Recall that $k$ determines the accuracy of the approximation of $\gamma_{n,i}^k$ (cf.~\eqref{naive}-\eqref{approx_g}), and therefore how accurately the update selection rule in~\eqref{update_rule} is implemented; at the same time, the computational complexity of implementing~\eqref{update_rule} increases with $k$ at a rate of $\mathcal{O}(p(k+1))$. Interestingly, experiments indicate that $k\ll{p}$ can be sufficient in practice, while sensitivity to $k$ only manifests itself for relatively small values of the compression ratio $d/D$. As seen in Fig.~\ref{fig:diff_ks}, RMSE of US-KF with $k=1$ is almost as low as the one achieved with $k=p$, while setting $k=0$ still yields reliable estimates, with the gap becoming smaller as $d/D$ increases. Recall that using $k=0$ leads to the simple rule in~\eqref{bit_naive}, and bears the additional advantage that no eigenpairs of $\mathbf{P}_{n|n-1,i}$ need be tracked.    
	 
	 \begin{table}[t]\label{tab:table2}
	 	\caption{{Average runtime of algorithms for $D=1000$.}}
	 	\vspace{-0.2 cm}
	 	\label{tab:table_2}
	 	\begin{center}
	 		\begin{tabular} {|c|c|c|c|}
	 			\hline
	 			$d/D$ & $0.05$ & $0.13$ & $0.24$ \\ \hline
	 			Random sampling & $0.16$ sec & $0.31$ sec & $0.81$ sec \\ \hline
	 			RP-KF & $0.26$ sec & $0.42$ sec & $1.3$ sec\\ \hline
	 			AC-KF & $0.41$ sec & $0.51$ sec & $1.05$ sec\\ \hline
	 			US-KF & $0.44$ sec & $0.64$ sec & $1.1$ sec\\ \hline 
	 			Full-data KF & $6.7$ sec & $6.7$ sec & $6.7$ sec\\ \hline
	 		\end{tabular}
	 	\end{center}
	 \end{table} 

	\subsection{Bud-KS}
	
In the last experiment, the extent to which backward smoothing iterations can improve reduced-observation filtering was examined. The AC-KF algorithm was first run with $d/D$ ranging from $0.0095$ up to $0.65$; Bud-KF was then run with $\tau_b=0$ in order to smooth all $N$ filtered estimates. Figure~\ref{fig:smoothed} depicts the average RMSE of the AC-KF with and without smoothing. Evidently, smoothing can significantly reduce RMSE over the entire range of dimensionality reduction, while its effect becomes more prominent as $d/D$ decreases. Upon examining Fig.~\ref{fig:smoothed}, the AC-KF using $<1\%$ of the data followed by Bud-KS, attains the same RMSE as the AC-KF using $5\%$ of the data; a surprising five-fold decrease. Thus, at the cost of introducing non-causality (or delay if a fixed-lag KS is used), smoothing offers room for significant decrease in the data requirements and complexity of tracking.
			\begin{figure}[t]
				\centering
				\centerline{\includegraphics[width=1.1\linewidth]{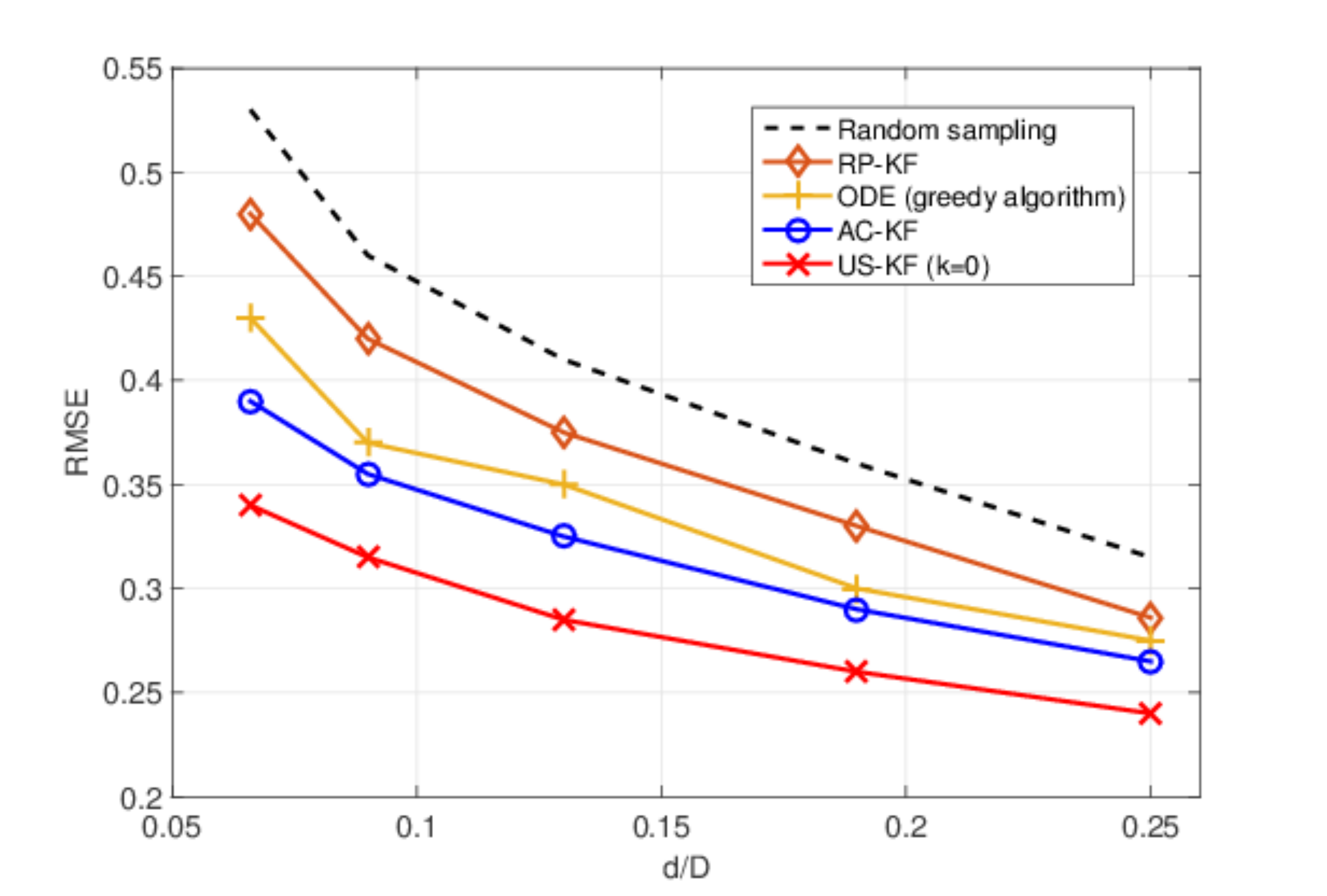}}
				\caption{Average RMSE for the US-KF, AC-KF, Greedy algorithm, RP-KF and random sampling as a function of $d/D$. }\label{fig:hsnr}
			\end{figure}
		\begin{figure}[t]
			\centering
			\centerline{\includegraphics[width=1.1\linewidth]{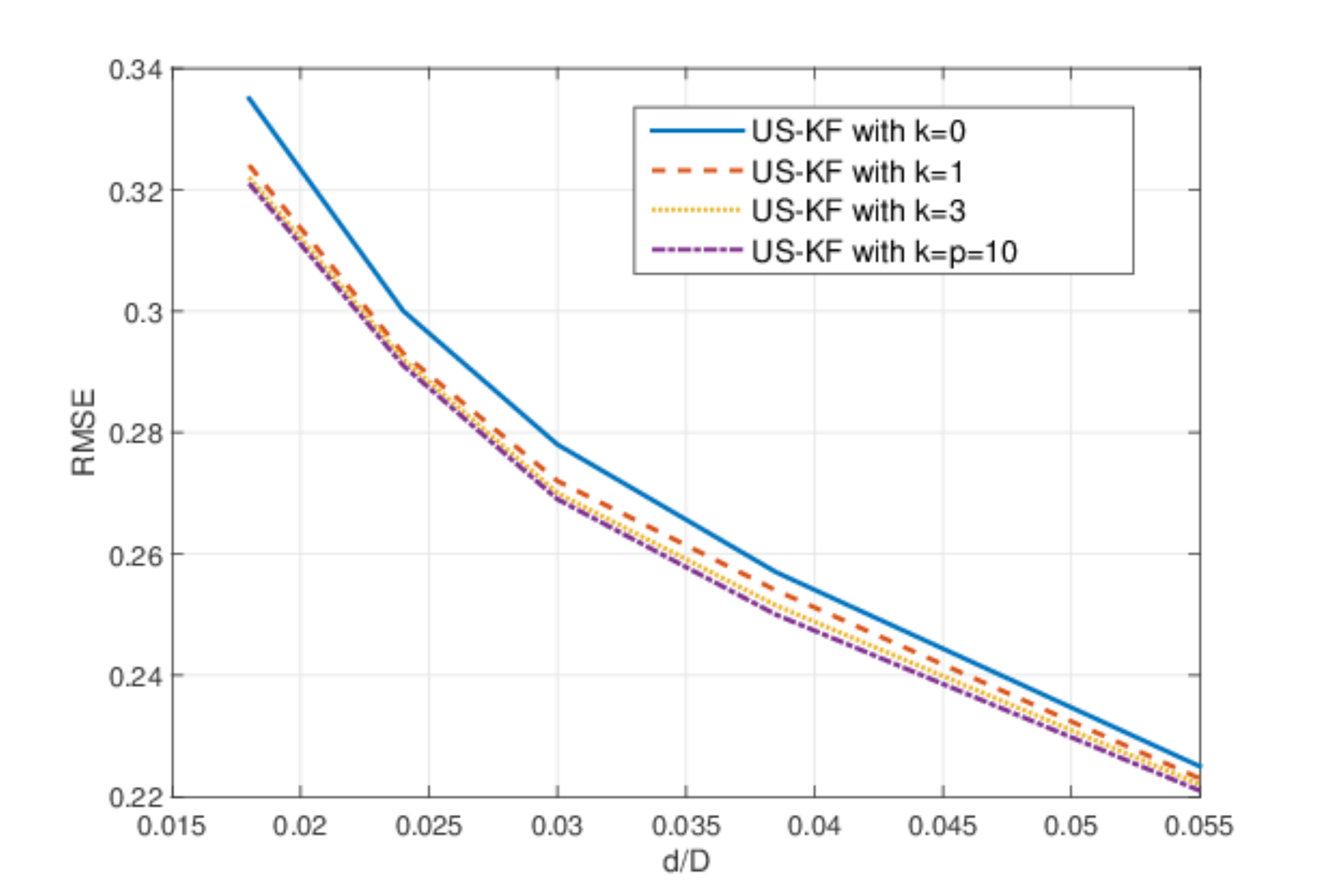}}
			\caption{Average RMSE for US-KF for different values of $k$.}\label{fig:diff_ks}
		\end{figure}
	\begin{figure}[t]
		\centering
		\centerline{\includegraphics[width=1.1\linewidth]{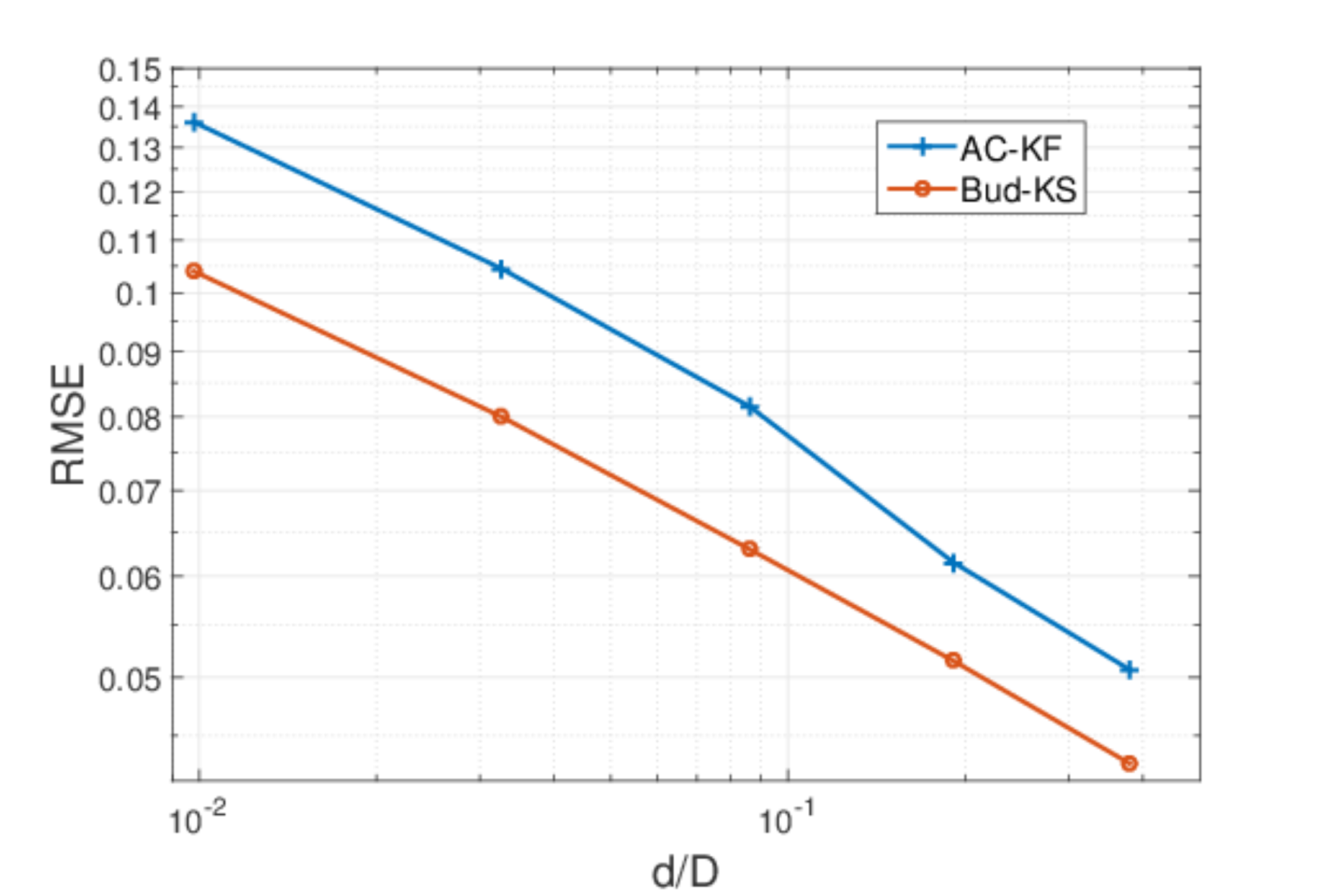}}
		\caption{RMSE of AC-KF versus Bud-KS, as a function of $d/D$.}\label{fig:smoothed}
	\end{figure}

\section{Application to monitoring dynamic graphs}\label{sec:networks}

 Dynamically evolving graphs offer a promising application domain for our proposed algorithms. In this context, measurements are obtained from a graph of known and constant topology in order to infer a set of hidden time-varying properties. Specifically, \emph{traffic matrix estimation} and \emph{link cost estimation} are two tasks that involve tracking of large-scale dynamical processes from linearly obtained observations. To demonstrate the applicability of US-KF in reducing the complexity of such tasks, a Kronecker graph $\mathcal{G}=(\mathcal{V},\mathcal{E})$ with $|\mathcal{V}|=50$ vertices  was generated. The adjacency matrix $\mathbf{A}_k$ of a Kronecker graph can be generated recursively as $\mathbf{A}_k=\mathbf{A}_{k-1}\otimes\mathbf{A}_{k-1}$, and is completely determined by the initiator graph $\mathbf{A}_{1}$. As shown in \cite{leskovec2010kronecker}, Kronecker graphs exhibit many real-word graph properties such as power-law degree distributions, and are thus highly recommended for simulating algorithms. For our experiments, a Kronecker graph was generated with initiator
\begin{align*}
\mathbf{A}_1=
\left[{\begin{array}{ccc}
	1 & 1 & 0 \\
	1 & 1 & 1\\
    0 & 1 & 1\\
	\end{array}}\right]
\end{align*}
until $100$ nodes become available. Nodes adjacent to all other nodes were removed in order to decrease the connectivity of the graph to more realistic levels. The resulting adjacency matrix is depicted in Fig.~\ref{fig:adj}. 

	\begin{figure}[t]
		\centering
		\centerline{\includegraphics[width=0.75\linewidth]{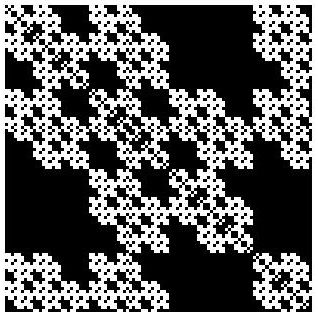}}
		\caption{Adjacency matrix of a Kronecker graph with $100$ nodes.}
		\label{fig:adj}
	\end{figure}

\subsection{Traffic matrix estimation}

Consider the task of measuring the traffic volume at the links of a network, in order to estimate the volume of origin-to-destination (OD) flows, a very important task in many networks ranging from the Internet to transportation. Since OD flows are defined by a set of origins $\mathcal{O}\subseteq{\mathcal{V}}$ and a set of destinations $\mathcal{D}\subseteq{\mathcal{V}}$, they can be represented as the entries an $|\mathcal{O}|\times|\mathcal{D}|$ traffic matrix $\mathbf{F}$. Similar to \cite{rajawat2014dynamic},~\cite{traffic1} and \cite{traffic2}, the following linear state-transition  and observation models is considered
\begin{align}\label{model_tme1}
\mathbf{f}_n&= \mathbf{f}_{n-1}+\mathbf{w}_n\\
\mathbf{l}_n&= \mathbf{R}\mathbf{f}_n + \mathbf{v}_n \label{model_tem2}
\end{align}
where $\mathbf{f}_n:=\mathrm{vec}(\mathbf{F}_n)$ is the vectorized traffic matrix at time slot $n$ that is assumed to evolve according to a random walk with driving Gaussian noise $\mathbf{w}_n$ with known covariance matrix $\sigma_f^2\mathbf{Q}$ such that $\mathbf{Q}_{i,j}=0.2^{|i-j|}$; $\mathbf{l}_n$ contains the link measurements at time slot $n$; and, $\mathbf{v}_n$ is the observation noise with $\mathrm{cov}(\mathbf{v}_n)=\sigma^2\mathbf{I}$. The choice of a non-diagonal $\mathbf{Q}_n$ was made to reflect the fact that flows tend to be highly correlated (see e.g.~\cite{kolaczyk2014statistical}). For this experiment, we set $\sigma_f=0.02$, $\sigma=0.5$, and generated the initial state as $\mathbf{f}_0\sim\mathcal{N}(2\cdot\mathbf{1},\mathbf{Q})$ In this model, the role of the measurement matrix is played by the \emph{routing matrix} $\mathbf{R}\in\{0,1 \}^{|\mathcal{E}|\times{|\mathcal{O}||\mathcal{D}|}}$, each column of which corresponds to an OD flow with entries taking the value $1$, if the corresponding links are part of the flow. Simply put, each column of $\mathbf{R}$ describes the path that the corresponding OD flow takes through the graph. For this experiment, OD paths were chosen to be the shortest possible using Dijkstra's algorithm. To make this experiment even more challenging, flows with paths that consist of a single link were not considered; flows with no sampled links and irrelevant links were also removed from the model. Overall, $189$ edges were sampled in order to track $689$ OD flows. 

Plotted in Fig.~\ref{fig:traffic} is the MSE ($\mathbb{E}[\|\mathbf{f}_n-\hat{\mathbf{f}}_n\|_2^2]$) of the estimated traffic matrix across time, for the proposed US-KF (Alg.\ref{rc_KF}), and the KF with random sampling. The algorithms were run for $N=100$ time slots and the results were averaged across $100$ runs. Both algorithms were tuned to utilize $6\%$ of edge measurements per time slot, and require approximately the same runtime. As seen in the plot, the estimates $\hat{\mathbf{f}}_n$ of the proposed US-KF converge faster than those of the sub-sampled KF, and keep a closer track of the true traffic matrix $\mathbf{f}_n$. It should be noted that, due to the large state dimension, other methods such as the RP-KF or greedy OED become impractical.

	\begin{figure}[t]
		\centering
		\centerline{\includegraphics[width=1.1\linewidth]{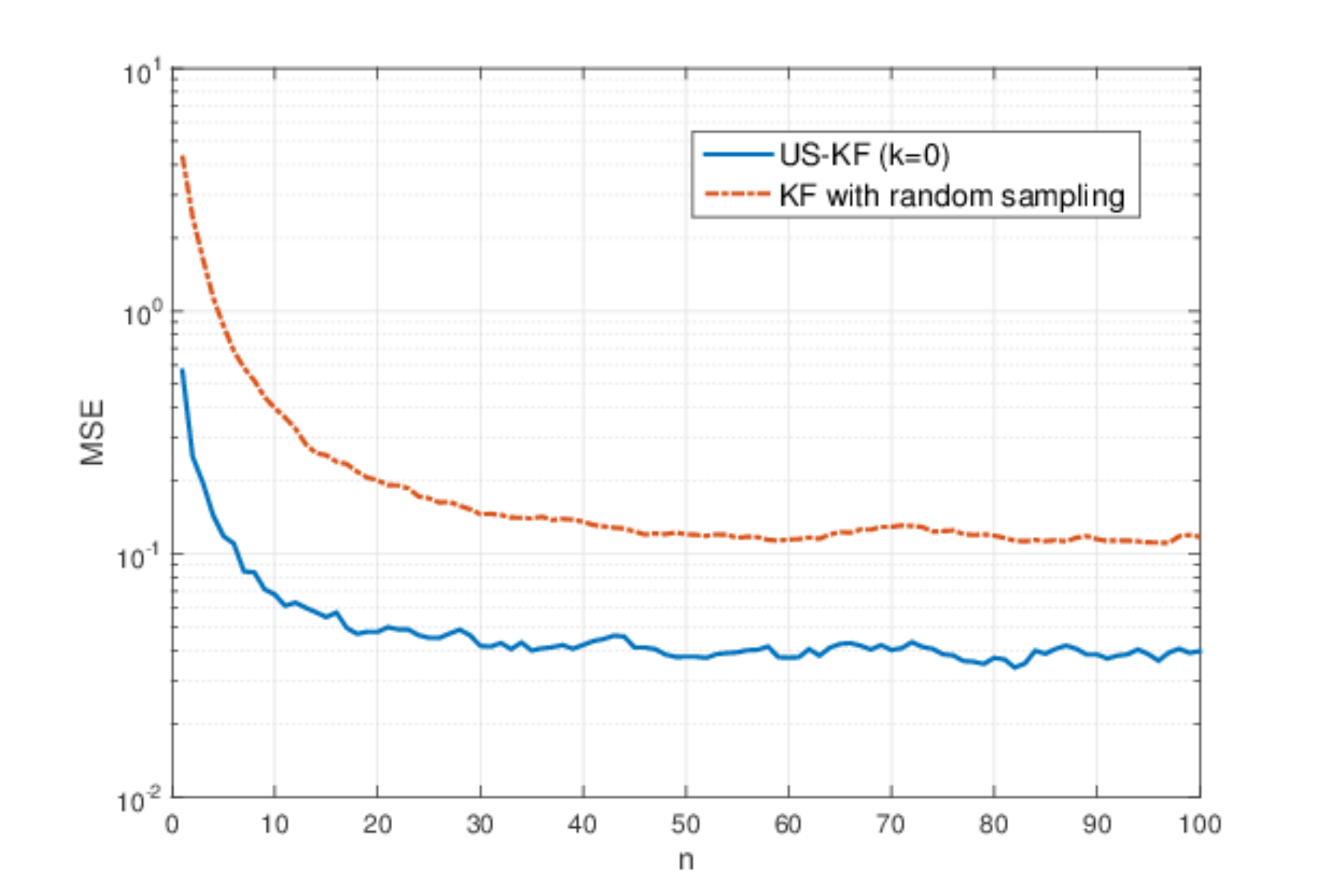}}
		\caption{Traffic matrix MSE vs time plot for the proposed Update selection KF (Alg.\ref{rc_KF}) and the random sampling KF. Both algorithms were tuned to utilize $6\%$ of edge measurements per time slot.}\label{fig:traffic}
	\end{figure}
		\begin{figure}[t]
			\centering
			\centerline{\includegraphics[width=1.1\linewidth]{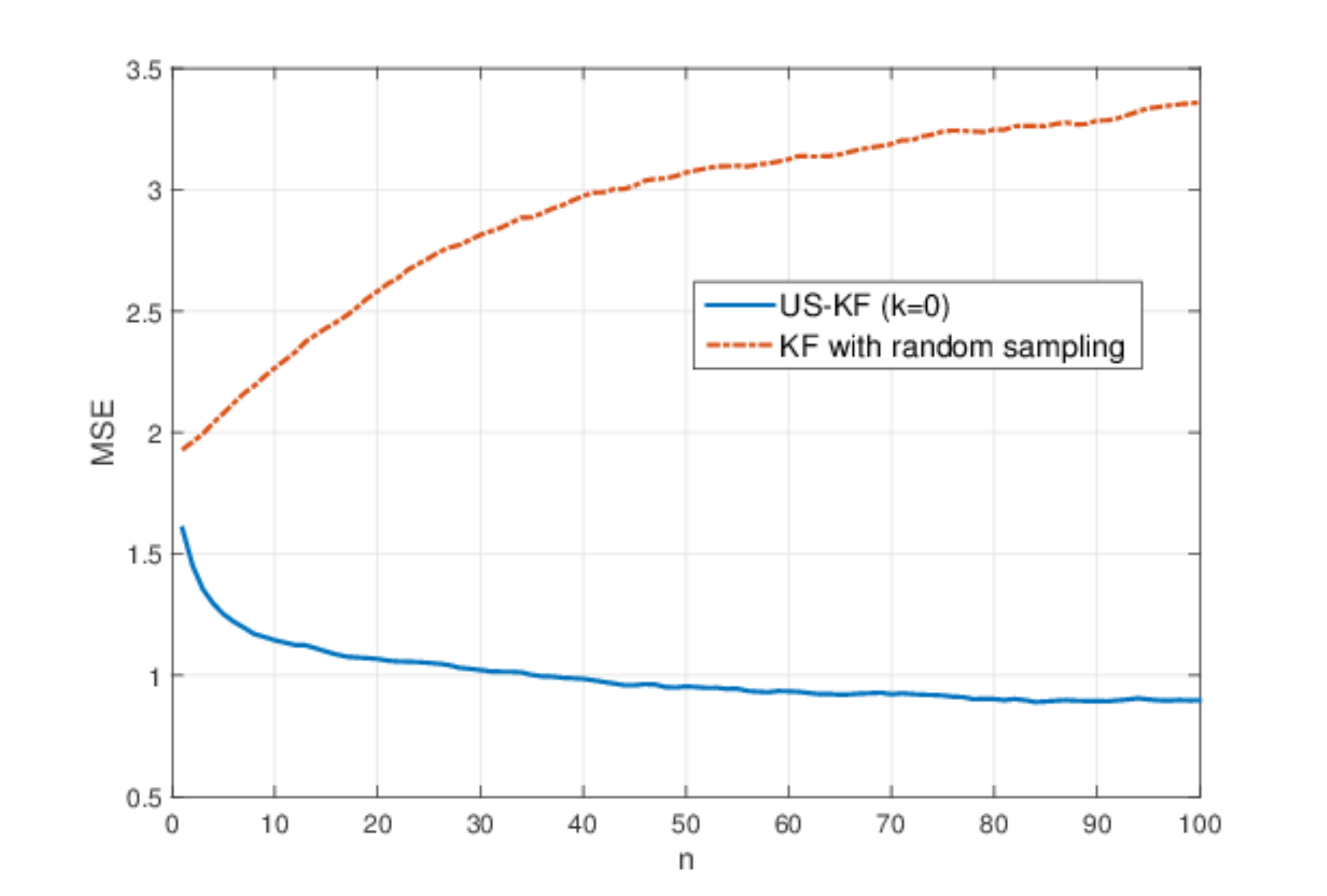}}
			\caption{Link-cost MSE vs time plot for US-KF (Alg.\ref{rc_KF}) and the random sampling KF. Both algorithms were tuned to utilize $4\%$ of flow measurements per time slot.}\label{fig:link}
		\end{figure}

\subsection{Estimation of link costs from path-cost measurements}

Consider now that every edge $\epsilon$ of the graph is associated with a cost $c(\epsilon)$, and that the concatenation of all such costs forms the link cost vector $\mathbf{c}$. A common task associated with networks is inference of $\mathbf{c}$ by measuring path costs $p_{ij}$, where $p_{ij}$ is the total cost of a flow between nodes $v_i$ and $v_j$ (see e.g.,~\cite[Chap. 9.4.1]{kolaczyk2014statistical}). Since $p_{ij}$ is the aggregation of all costs of the edges that the corresponding path crosses, it can be expressed as the inner product between $\mathbf{c}$ and the corresponding row of the routing matrix. Consequently, path costs and link costs are linked through the linear observation model $\mathbf{p}=\mathbf{R}^T\mathbf{c}$, where $\mathbf{p}$ is the vector with all the available path cost measurements. Considering dynamic graphs where the link costs $\mathbf{c}_n$ and path costs $\mathbf{p}_n$ evolve across time slots $n$, leads to the familiar linear state-transition and state-observation models
\begin{align}\label{model_link1}
\mathbf{c}_n&= \mathbf{c}_{n-1}+\mathbf{w}_n\\
\mathbf{p}_n&= \mathbf{R}^T\mathbf{c}_n + \mathbf{v}_n \label{model_link2}
\end{align}
where $\mathbf{w}\sim{\mathcal{N}(\mathbf{0},\sigma_c^2\mathbf{I})}$,  $\mathbf{v}\sim{\mathcal{N}(\mathbf{0},\sigma^2\mathbf{I})}$, and the initial state is $\mathbf{c}_0\sim{\mathcal{N}(\mathbf{m},\sigma_0^2\mathbf{I})}$. For this experiment, we used the same graph and routing matrix as in the traffic estimation experiment, and generated $\mathbf{c}_n$ and $\mathbf{p}_n$ according to~\eqref{model_link1} and~\eqref{model_link2} correspondingly, with $\sigma_c=0.04$, $\sigma=0.1$, $\mathbf{m}=\mathbf{1}$ and $\sigma_0=0.1$.

Plotted in Fig.~\ref{fig:link} is the MSE ($\mathbb{E}[\|\mathbf{c}_n-\hat{\mathbf{c}}_n\|_2^2]$) of the estimated link costs across time, for the proposed US-KF (Alg.\ref{rc_KF}) and the KF with random sampling, for $N=100$ time slots and averaged across $100$ runs. Both algorithms were tuned to utilize $4\%$ of path cost measurements per time slot, and require approximately the same runtime. As seen in the plot, the proposed US-KF successfully tracks the slowly evolving link costs by judiciously selecting and using a small fraction of the available path cost observations. Furthermore, it can be observed that if the same fraction ($4\%$) of measurements is selected at random, then the KF fails to track the link costs, with its estimate diverging from the true value as time progresses. The divergence of the KF with random sampling is consistent with the results in~\cite{sinopoli2004kalman}, where it is shown that there exists a cut-off value for the data rate, below which the error covariance may become unbounded. Interestingly, the proposed reduced-complexity US-KF appears to be much more robust to divergence; as discussed in the following remark. 

{\textbf{Remark 4:} While KF based on random sampling (as well RP-KF) diverges when the compression ratio $d/D$ becomes smaller than a certain threshold, this is not the case for the advocated censoring-based alternatives (AC-KF and US-KF) since diverging estimates prohibit censoring. This becomes evident upon realizing that a diverging estimate (i.e., $\|\hat{\boldsymbol{\theta}}_n-\boldsymbol{\theta}_n\|_2\rightarrow \infty$) would imply infinitely large innovations that cannot be smaller than finite thresholds such as the ones used in censoring rules \eqref{entry_censoring} and \eqref{update_rule}.  This in turn implies that if AC-KF and US-KF were divergent, they would become equivalent to the full data KF.  In a nutshell, if the ordinary KF is not divergent, the same holds for the proposed AC-KF and US-KF, since the latter will always obtain sets of observations that guarantee a bounded tracking error. }

\section{Concluding Remarks}\label{sec:conclusion}

We introduced random projections and censoring as dimensionality reduction and measurement selection methods for tracking dynamical processes with generally time-varying parameters. The proposed methods are simple routines that can be used as dimensionality reduction modules coupled with an ordinary KF. Furthermore, we introduced a reduced-complexity KF that processes measurements sequentially and performs updates that are deemed informative based on the information gain of corresponding measurements.  Performance was not analytically performed, but simulations provide surprisingly strong evidence that the proposed methods perform close to the greedy measurement selection method in terms of estimation error. Furthermore, censoring-based measurement selection enjoys much lower computational complexity than greedy OED. To demonstrate applicability of the proposed update selection approach on real-world problems, we examined the network-related applications of traffic matrix estimation and network flow estimation.  

\appendix
\begin{IEEEproof}[Proof of Proposition~\ref{pro:rps}]
	
Follows readily from Theorem 2 in~\cite{drineas2011faster}.	

\end{IEEEproof}

\begin{IEEEproof}[Proof of Proposition~\ref{pro:bias}]

From the assumption of large and uncorrelated noise $\mathbf{R}_n=\sigma_n^2\mathbf{I}$, the inverse reduced innovation covariance matrix can be approximated as
\begin{equation*}
\left(\check{\mathbf{X}}_n\mathbf{P}_{n|n-1}\check{\mathbf{X}}_n^T+\check{\mathbf{R}}_n\right)^{-1}\approx{\sigma_n^{-2}\mathbf{I}}
\end{equation*}
and hence the correction update as 
\begin{align}\label{this1}
\hat{\boldsymbol{\theta}}_{n|n}=\hat{\boldsymbol{\theta}}_{n|n-1}+\sigma_n^{-2}\mathbf{P}_{n|n-1}\sum_{i=1}^D\mathbf{x}_{n,i}\tilde{y}_{n,i}(1-c_{n,i})
\end{align}
where $\tilde{y}_{n,i}:=\mathbf{x}_{n,i}^T({\boldsymbol{\theta}}_{n}-\hat{\boldsymbol{\theta}}_{n|n-1})+v_{n,i}$. Furthermore, for $\mu=0$ the censoring rule in~\eqref{cis} simplifies to $1-c_{n,i}=\mathbbm{1}_{|\{\tilde{y}_{n,i}|\geq{\tau_n\sigma_{n}}\}}$, where $\tilde{y}_{n,i}\sim{\mathcal{N}(0,\mathbf{x}_{n,i}^T\mathbf{P}_{n|n-1,i}\mathbf{x}_{n,i}+\sigma_n^2})$. If $\hat{\boldsymbol{\theta}}_{n-1|n-1}$ is unbiased, then it readily follows that $\hat{\boldsymbol{\theta}}_{n|n-1}$ is also unbiased, and~\eqref{this1} yields
\begin{align}\label{this2}
\mathbb{E}[\hat{\boldsymbol{\theta}}_{n|n}-\boldsymbol{\theta}_n]=\sigma_n^{-2}\mathbf{P}_{n|n-1}\sum_{i=1}^D\mathbf{x}_{n,i}\mathbb{E}[\tilde{y}_{n,i}(1-c_{n,i})].
\end{align}
Since 
\begin{align}\nonumber
\mathbb{E}[\tilde{y}_{n,i}(1-c_{n,i})]&=\mathbb{E}[\tilde{y}_{n,i}\mathbbm{1}_{|\{\tilde{y}_{n,i}|\geq{\tau_n\sigma_{n}}\}}]\\\nonumber
&=\mathbb{E}[\tilde{y}_{n,i}]-\mathbb{E}[\tilde{y}_{n}\mathbbm{1}_{|\{\tilde{y}_{n,i}|\leq{\tau_n\sigma_{n,i}}\}}]\\\nonumber
&\propto-\int_{-\tau_n\sigma_{n}}^{\tau_n\sigma_{n}}\tilde{y}_{n,i}e^{-c(\tilde{y}_{n,i})^2}d({\tilde{y}_{n,i}})\\\label{this3}
&=\frac{1}{2c}(e^{-c\tau_n^2\sigma_{n}^2}-e^{-c(-\tau_n\sigma_{n})^2})=0
\end{align} 
 where $c:=0.5(\mathbf{x}_{n,i}^T\mathbf{P}_{n|n-1,i}\mathbf{x}_{n,i}+\sigma_n^2)^{-1}$, it follows from~\eqref{this3} and~\eqref{this2} that $\mathbb{E}[\hat{\boldsymbol{\theta}}_{n|n}-\boldsymbol{\theta}_n]=\mathbf{0}$, and the AC-KF is unbiased. 

\end{IEEEproof}

\begin{IEEEproof}[Proof of Proposition~\ref{pro:dkl}]

For observations generated according to the linear Gaussian model, and since $\hat{\boldsymbol{\theta}}_{n|n,i}$ is the MMSE estimator of $\boldsymbol{\theta}_n$ given $\hat{\boldsymbol{\theta}}_{n|n-1}$ and $\mathbf{y}_{n,{1:i}}$, it follows that the posterior of $\boldsymbol{\theta}_n$ is also Gaussian with $p_{n,i}({\boldsymbol{\theta}}_{n})=\mathcal{N}\left(\hat{\boldsymbol{\theta}}_{n|n,i},\mathbf{P}_{n|n,i}\right)$. Similarly, one can obtain $p_{n,i-1}({\boldsymbol{\theta}}_{n})=\mathcal{N}\left(\hat{\boldsymbol{\theta}}_{n|n,i-1},\mathbf{P}_{n|n,i-1}\right)$. Using the closed-form identity for the KL divergence between two multivariate normal pdfs, we arrive at
\begin{align}\nonumber
 \mathcal{D}_{KL}(p_{n,i}|| p_{n,i-1})&=\frac{1}{2}\bigg[\mathrm{tr}\left(\mathbf{P}_{n|n,i-1}^{-1}\mathbf{P}_{n|n,i}\right)\\\nonumber
 &+ (\hat{\boldsymbol{\theta}}_{n|n,i-1}-\hat{\boldsymbol{\theta}}_{n|n,i})^T\mathbf{P}_{n|n,i-1}^{-1}\\\nonumber
 &\times(\hat{\boldsymbol{\theta}}_{n|n,i-1}-\hat{\boldsymbol{\theta}}_{n|n,i})\\\label{ayto_edw}
 &-p + \ln\left(\frac{|\mathbf{P}_{n|n,i-1}|}{|\mathbf{P}_{n|n,i}|}\right)\bigg]
\end{align}
where $\mathrm{tr}(\mathbf{P})$ denotes the trace of matrix $\mathbf{P}$ and $|\mathbf{P}|$ its determinant.

Using \eqref{covar_update}, the first summand in~\eqref{ayto_edw} can be expressed as
\begin{align}\nonumber
\mathrm{tr}&\left(\mathbf{P}_{n|n,i-1}^{-1}\mathbf{P}_{n|n,i}\right)\\\nonumber
&=\mathrm{tr}\left(\mathbf{I}_p-\mathbf{x}_{n,i}\mathbf{x}_{n,i}^T\mathbf{P}_{n|n,i-1}s_{n,i}^{-1}\right)\\\label{tr}
&=p-\mathbf{x}_{n,i}^T\mathbf{P}_{n|n,i-1}\mathbf{x}_{n,i}s_{n,i}^{-1}.
\end{align}	

Upon observing that for the RLS-like iteration in \eqref{par_update} the inverse of the covariance matrix is updated as
\begin{align}\label{inv_cov}
 \mathbf{P}_{n|n,i}^{-1}=\mathbf{P}_{n|n,i-1}^{-1}+\mathbf{x}_{n,i}\mathbf{x}_{n,i}^T\sigma_i^{-2}
 \end{align}
  the fourth summand in~\eqref{ayto_edw} can be expressed as
\begin{align}\nonumber
\ln\left(\frac{|\mathbf{P}_{n|n,i-1}|}{|\mathbf{P}_{n|n,i}|}\right)&=\ln\left({|\mathbf{P}_{n|n,i-1}|}\right)+\ln\left({|\mathbf{P}_{n|n,i}^{-1}|}\right)\\\nonumber
&=\ln\left({|\mathbf{P}_{n|n,i-1}|}\right)\\ \nonumber
&+\ln\left({|\mathbf{P}_{n|n,i-1}^{-1}+\mathbf{x}_{n,i}\mathbf{x}_{n,i}^T\sigma_i^{-2}|}\right)\\\nonumber
&=\ln\left({|\mathbf{P}_{n|n,i-1}|}\right)\\ \nonumber
&+\ln\left({|\mathbf{P}_{n|n,i-1}^{-1}+\mathbf{x}_{n,i}\mathbf{x}_{n,i}^T\sigma_i^{-2}|}\right)\\ \nonumber
&=\ln\left({|\mathbf{P}_{n|n,i-1}|}\right)\\ \nonumber
&+\ln\left(|\mathbf{P}_{n|n,i-1}^{-1}|\left(1+\mathbf{x}_{n,i}^T\mathbf{P}_{n|n,i-1}\mathbf{x}_{n,i}\sigma_i^{-2}\right)\right)\\ \nonumber
&=\ln\left(1+\mathbf{x}_{n,i}^T\mathbf{P}_{n|n,i-1}\mathbf{x}_{n,i}\sigma_i^{-2}\right)\\\label{fourth}
&=\ln(s_{n,i})-\ln(\sigma_i^2)
\end{align}	
where in the first equality we used the fact that $|\mathbf{P}^{-1}|=1/|\mathbf{P}|$, and in the fourth one we applied the matrix determinant lemma for rank-one updates.
	
Finally, since $\hat{\boldsymbol{\theta}}_{n|n,i-1}-\hat{\boldsymbol{\theta}}_{n|n,i}=-\mathbf{k}_{n,i}e_{n,i}$, the second summand in~\eqref{ayto_edw} becomes
\begin{align}\nonumber
(\hat{\boldsymbol{\theta}}_{n|n,i-1}-\hat{\boldsymbol{\theta}}_{n|n,i})^T&\mathbf{P}_{n|n,i-1}^{-1}(\hat{\boldsymbol{\theta}}_{n|n,i-1}-\hat{\boldsymbol{\theta}}_{n|n,i})\\ \nonumber
&=(\mathbf{k}_{n,i}e_{n,i})^T\mathbf{P}_{n|n,i-1}^{-1}\mathbf{k}_{n,i}e_{n,i}\\ \label{second}
&=e_{n,i}^2\mathbf{x}_{n,i}^T\mathbf{P}_{n|n,i-1}\mathbf{x}_{n,i}s_{n,i}^{-2}.
\end{align}	
	
Substituting \eqref{tr}-\eqref{second} into~\eqref{ayto_edw} and with $\gamma_{n,i}=\mathbf{x}_{n,i}^T\mathbf{P}_{n|n,i-1}\mathbf{x}_{n,i}$, we arrive at the result of Proposition \ref{pro:dkl}.		
\end{IEEEproof}

\begin{IEEEproof}[Proof of Proposition~\ref{pro:sym_dkl}]
	By the definition of $ \mathcal{D}(p_{n,i}, p_{n,i-1})$ in \eqref{sym_dkl} and expressing $ \mathcal{D}_{KL}(p_{n,i-1}|| p_{n,i})$ using arguments similar to \eqref{tr} and \eqref{fourth}, it follows that 
	\begin{align}\nonumber
    \mathcal{D}(p_{n,i}, p_{n,i-1})&=\frac{1}{2}(\hat{\boldsymbol{\theta}}_{n|n,i-1}-\hat{\boldsymbol{\theta}}_{n|n,i})^T\\\nonumber
    &\times\left(\mathbf{P}_{n|n,i-1}^{-1}+\mathbf{P}_{n|n,i}^{-1}\right)\\
    &\times(\hat{\boldsymbol{\theta}}_{n|n,i-1}-\hat{\boldsymbol{\theta}}_{n|n,i}).
	\end{align}	
	Utilizing \eqref{inv_cov} and that $\hat{\boldsymbol{\theta}}_{n|n,i-1}-\hat{\boldsymbol{\theta}}_{n|n,i}=-\mathbf{k}_{n,i}e_{n,i}$ yields
	
	\begin{align}\nonumber
	\mathcal{D}(p_{n,i}, p_{n,i-1})&=\frac{1}{2}e_{n,i}^2\mathbf{k}_{n,i}^T\\ \nonumber
	&\times\left(2\mathbf{P}_{n|n,i-1}^{-1}+\mathbf{x}_{n,i}\mathbf{x}_{n,i}^T\sigma_i^{-2}\right)\mathbf{k}_{n,i} \\ \nonumber
	&=\frac{1}{2}\frac{e_{n,i}^2}{s_{n,i}^2}\bigg(\mathbf{x}_{n,i}^T\mathbf{P}_{n|n,i-1}\mathbf{x}_{n,i}\\\nonumber
	&+(\mathbf{x}_{n,i}^T\mathbf{P}_{n|n,i-1}\mathbf{x}_{n,i})^2\sigma_i^{-2}\bigg)
	\end{align}	
	and since $\gamma_{n,i}:=\mathbf{x}_{n,i}^T\mathbf{P}_{n|n,i-1}\mathbf{x}_{n,i}$ and $\bar{e}_{n,i}:=e_{n,i}s_{n,i}^{-1/2}$ the proposition holds.
\end{IEEEproof}	

\begin{IEEEproof}[Proof of Proposition~\ref{pro:delta_mse}]
  
  Recalling that $e_{n,i}:=y_{n,i}-\mathbf{x}_{n,i}^T\hat{\boldsymbol{\theta}}_{n|n,i-1}$ and $y_{n,i}=\mathbf{x}_{n,i}^T\boldsymbol{\theta}_n +v_{n,i}$, \eqref{1sto_update}	 yields
  	\begin{align}\nonumber
\hat{\boldsymbol{\theta}}_{n|n,i}&=\hat{\boldsymbol{\theta}}_{n|n,i-1}-\mu_{n,i}\mathbf{x}_{n,i}\mathbf{x}_{n,i}^T(\hat{\boldsymbol{\theta}}_{n|n,i-1}-\boldsymbol{\theta}_n)\\\label{theta}
&+\mu_{n,i}\mathbf{x}_{n,i}v_{n,i}.
  	\end{align}
  With $\tilde{\boldsymbol{\theta}}_{n,i}:=\hat{\boldsymbol{\theta}}_{n|n,i}-\boldsymbol{\theta}_n$ denoting the error vector,~\eqref{theta} can be expressed as
    	\begin{align}\label{zeta}
    	\tilde{\boldsymbol{\theta}}_{n,i}=\left( \mathbf{I}_p- \mu_{n,i}\mathbf{x}_{n,i}\mathbf{x}_{n,i}^T \right)\tilde{\boldsymbol{\theta}}_{n,i-1}+\mu_{n,i}\mathbf{x}_{n,i}\mathbf{v}_{n,i}. 
    	\end{align}
  The outer product of both sides in \eqref{zeta} yields
      	\begin{align}\nonumber
      	\tilde{\boldsymbol{\theta}}_{n,i}\tilde{\boldsymbol{\theta}}_{n,i}^T&=\left( \mathbf{I}_p- \mu_{n,i}\mathbf{x}_{n,i}\mathbf{x}_{n,i}^T \right)\tilde{\boldsymbol{\theta}}_{n,i-1}\\\nonumber
      	&\times\tilde{\boldsymbol{\theta}}_{n,i-1}^T\left( \mathbf{I}_p- \mu_{n,i}\mathbf{x}_{n,i}\mathbf{x}_{n,i}^T \right)\\\nonumber
      	&+2\left( \mathbf{I}_p- \mu_{n,i}\mathbf{x}_{n,i}\mathbf{x}_{n,i}^T \right)\tilde{\boldsymbol{\theta}}_{n,i-1}\mu_{n,i}\mathbf{x}_{n,i}^T\mathbf{v}_{n,i}\\\label{full}
      	&+(\mu_{n,i})^2\mathbf{x}_{n,i}\mathbf{x}_{n,i}^T(\mathbf{v}_{n,i})^2.
      	\end{align}
  Since $\hat{\boldsymbol{\theta}}_{n|n,i-1}$ is unbiased, it follows that  $\hat{\boldsymbol{\theta}}_{n|n,i}$ is unbiased too, and therefore the MSE equals the trace of the covariance matrix. Since the expected value of the second summand in \eqref{full} is zero, the trace of the expectation in \eqref{full} yields
        	\begin{align}\nonumber
       \mathrm{tr}\left(\bar{\mathbf{P}}_{n|n,i}\right)&=\mathrm{tr}\left(\left( \mathbf{I}_p- \mu_{n,i}\mathbf{x}_{n,i}\mathbf{x}_{n,i}^T \right)^2\mathbf{P}_{n|n,i-1}\right)\\\nonumber
       &+\mu_{n,i}^2\|\mathbf{x}_{n,i}\|_2^2\sigma_i^2\\\nonumber
       &=\mathrm{tr}\left(\mathbf{P}_{n|n,i-1}\right)+\mu_{n,i}^2\mathrm{tr}\left((\mathbf{x}_{n,i}\mathbf{x}_{n,i}^T)^2\mathbf{P}_{n|n,i-1}\right)\\\label{last}
       &-2\mu_{n,i}\mathrm{tr}\left(\mathbf{x}_{n,i}\mathbf{x}_{n,i}^T\mathbf{P}_{n|n,i-1}\right)
    +\mu_{n,i}^2\|\mathbf{x}_{n,i}\|_2^2\sigma_i^2.
        	\end{align}
        	where $\bar{\mathbf{P}}_{n|n,i}$ is the covariance matrix after the first-order update in \eqref{1sto_update}. Given that $\Delta_n(\mu_{n,i}):=\mathrm{tr}\left(\bar{\mathbf{P}}_{n|n,i}\right)-\mathrm{tr}\left(\mathbf{P}_{n|n,i-1}\right)$, and upon observing that $\mathrm{tr}\left(\mathbf{x}_{n,i}\mathbf{x}_{n,i}^T\mathbf{P}_{n|n,i-1}\right)=\gamma_{n,i}$ and $\mathrm{tr}\left((\mathbf{x}_{n,i}\mathbf{x}_{n,i}^T)^2\mathbf{P}_{n|n,i-1}\right)=\|\mathbf{x}_{n,i}\|_2^2\gamma_{n,i}$, the proof is complete after using \eqref{last}.
\end{IEEEproof}	

\bibliographystyle{IEEEtran}
\bibliography{ref_censoring}
\end{document}